\documentclass[USenglish,oneside,twocolumn]{article}

\usepackage[utf8]{inputenc}
\usepackage[big]{dgruyter_NEW}
\usepackage{graphicx}
\usepackage{subfig}
\usepackage{balance}
\usepackage{amssymb}
\usepackage{epstopdf}
\usepackage{listings}
\usepackage{color}
\definecolor{keywordcolor}{rgb}{0.8,0.1,0.5}
\definecolor{webgreen}{rgb}{0,.5,0}
\usepackage{floatrow,blindtext}
\usepackage{enumitem} 
\usepackage[nice]{nicefrac}
\usepackage{times}
\usepackage{multirow}
\usepackage{epsfig}
\usepackage{amsmath}
\usepackage{amssymb}
\usepackage[titletoc]{appendix}
\usepackage{amssymb}
\usepackage{epstopdf}
\usepackage{listings}
\usepackage{color}
\definecolor{keywordcolor}{rgb}{0.8,0.1,0.5}
\definecolor{webgreen}{rgb}{0,.5,0}
\usepackage{floatrow,blindtext}
\usepackage{enumitem} 
\usepackage[nice]{nicefrac}
\usepackage{times}
\usepackage{epsfig}
\usepackage[titletoc]{appendix}
\usepackage{graphicx}
   \graphicspath{{./figs/}}
\usepackage{booktabs}
\usepackage{amsmath}
\usepackage{amssymb}
\usepackage{datetime}
\usepackage{bbm}
\usepackage{wasysym}
\usepackage{balance}
\usepackage{color,marvosym,xspace}
\usepackage{hyperref}
\usepackage{url}
\usepackage{amsmath}
\usepackage{listings}
\usepackage{relsize}
\usepackage[svgnames]{xcolor}
\usepackage{diagbox}
\usepackage{units}
\usepackage{enumitem}
\usepackage{enumitem}
\usepackage{ifsym}
\usepackage{color,marvosym,xspace}
\usepackage{moreverb}
\usepackage{breakurl}
\usepackage{pifont}
\usepackage{natbib}
\usepackage{float}
\usepackage{enumitem}
\usepackage{ifsym}
\usepackage{hyperref}
\usepackage{bookmark}
\usepackage{times}
\usepackage{helvet}
\usepackage{courier}
\usepackage{epsfig}
\usepackage[titletoc]{appendix}
\usepackage{graphicx}
\usepackage{floatrow}
   \graphicspath{{./figs/}}
\usepackage{booktabs}
\usepackage{amsmath}
\usepackage{amssymb}
\usepackage{bbm}
\usepackage{wasysym}
\usepackage{balance}
\usepackage{color,marvosym,xspace}
\usepackage[hyphens]{url}
\usepackage{url}
\usepackage{float}
\usepackage{enumitem}
\usepackage{ifsym}
\usepackage{scrextend}
\usepackage{amsfonts}
\usepackage{amsthm}
\usepackage{slashbox}
\usepackage{pict2e}

\newtheorem{Thm}{Theorem}
\newtheorem{defn}{Definition}

\newcommand{\dpautogm}{DP-AuGM\xspace}
\newcommand{\dpdl}{DP-DL\xspace}
\newcommand{\dpvae}{DP-VaeGM\xspace}
\newcommand{\hospital}{Hospital Data\xspace}
\newcommand{\mnist}{MNIST\xspace}
\newcommand{\adult}{Adult Census Data\xspace}
\newcommand{\malware}{Malware Data\xspace}
\newcommand{\spate}{sPATE\xspace}

\newcommand{\hmark}{\textcolor{blue}{\ding{51}}\textsuperscript{\textcolor{blue}{\kern-1em\ding{55}}}}

\makeatletter

\newcommand{\Rmnum}[1]{\expandafter\@slowromancap\romannumeral #1@}

\makeatother

\cclogo{\includegraphics{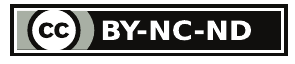}}

\begin{document}

  \author[1]{Qingrong Chen}

  \author[2]{Chong Xiang}

  \author[3]{Minhui Xue}

  \author[4]{Bo Li}
  
  \author[5]{Nikita Borisov}
  
  \author[6]{Dali Kaafar}
  
  \author[7]{Haojin Zhu}

  \affil[1]{University of Illinois at Urbana-Champaign, Email: qc16@illinois.edu;}

  \affil[2]{Shanghai Jiao Tong University, Email: danco2015@sjtu.edu.cn;}

  \affil[3]{Macquarie University and Data61-CSIRO, Email: minhui.xue@mq.edu.au;}

  \affil[4]{University of Illinois at Urbana-Champaign, Email: lxbosky@gmail.com;}
  
  \affil[5]{University of Illinois at Urbana-Champaign, Email: nikita@illinois.edu;}
  
   \affil[6]{Macquarie University and Data61-CSIRO, Email: dali.kaafar@mq.edu.au;}

 \affil[7]{Shanghai Jiao Tong University, Email: zhuhaojin@gmail.com.}

  \title{\huge Differentially Private Data Generative Models}

  \runningtitle{Differentially Private Data Generative Models}

  \begin{abstract}
{Deep neural networks (DNNs) have recently been widely adopted in various applications, and such success is largely due to a combination of algorithmic breakthroughs, computation resource improvements, and access to a large amount of data. 
However, the large-scale data collections required for deep learning often contain sensitive information, therefore raising many privacy concerns.
Prior research has shown several successful attacks in inferring sensitive training data information, such as model inversion~\cite{usenix2014,ccs15}, membership inference~\cite{shokri2016membership}, and generative adversarial networks (GAN) based leakage attacks against collaborative deep learning~\cite{hitaj2017deep}. In this paper, to enable learning efficiency as well as to generate data with privacy guarantees and high utility, we propose a \textbf{d}ifferentially \textbf{p}rivate \textbf{au}toencoder-based \textbf{g}enerative \textbf{m}odel (\dpautogm) and a \textbf{d}ifferentially \textbf{p}rivate \textbf{va}riational auto\textbf{e}ncoder-based \textbf{g}enerative \textbf{m}odel (\dpvae). We  evaluate the robustness of two proposed models. We show that \dpautogm can effectively defend against the model inversion, membership inference, and GAN-based attacks. We also show that \dpvae is robust against the membership inference attack. We conjecture that the key to defend against the model inversion and GAN-based attacks is not due to differential privacy but the perturbation of training data. Finally, we demonstrate that both \dpautogm and \dpvae can be easily integrated with real-world machine learning applications, such as \emph{machine learning as a service} and \emph{federated learning}, which are otherwise threatened by the membership inference attack and the GAN-based attack, respectively. }
\end{abstract}

\keywords{Differential privacy; Generative models; Robustness}

\maketitle
\section{Introduction}
Advanced machine learning techniques, and in particular deep neural networks (DNNs), have been applied with great success to a variety of areas, including speech processing~\cite{hinton2012deep}, medical diagnostics~\cite{ciresan2012deep}, image processing~\cite{ciregan2012multi}, and robotics~\cite{zhang2015towards}.
Such success largely depends on massive collections of data for training machine learning models.
However, these data collections often contain sensitive information and therefore raise many privacy concerns. Several privacy violation attacks have been proposed to show that it is possible to extract sensitive and private information from different learning systems. 
Specifically, Fredrikson et al.~\cite{usenix2014} proposed to infer sensitive patients' genomic markers by actively probing the outputs from the model and auxiliary demographic information about them. In a follow-up study, Fredrikson et al.~\cite{ccs15} developed a more robust model inversion attack using predicted confidence values to recover confidential information of a training set (e.g., human faces). Shokri and Shmatikov~\cite{shokri2016membership} proposed a membership inference attack, which tries to predict whether a data point belongs to the training set. More recently, a generative adversarial network (GAN) based attack against collaborative deep learning~\cite{hitaj2017deep} was proposed against distributed machine learning systems, where users collaboratively train a model by sharing gradients of their locally trained models through a parameter server. The GAN-based attack has shown that even when the training process is differentially private~\cite{abadi2016deep,papernot2016semi}, it is still possible to mount an attack to extract sensitive information from original training data~\cite{hitaj2017deep} as trusted servers may leak information unintentionally. Given the fact that Google has proposed federated learning based on distributed machine learning~\cite{mcmahan2017federated} and has already deployed it to mobile devices, such a GAN based attack~\cite{hitaj2017deep} raises serious privacy concerns.

In this paper, we propose to use differentially private data generative models to publish differentially private synthetic data that can both protect privacy and retain high data utility. Such data generative models are trained over private/sensitive data (we will denote it as private data to be aligned with the definition in~\cite{papernot2018scalable}) in a differentially private manner, and are able to generate new surrogate data for later learning tasks. As a result, the generated data preserves the statistical properties of the private data, which enables high learning efficacy, while also protecting the privacy of the private data. The approach of using differentially private data generative models has several advantages. \textit{First}, with the generative models, privacy can be preserved even if the entire trained model or the generated data is accessible to an adversary. \textit{Second}, it can be easily integrated with other learning tasks without adding much overhead, since only the training data is changed. \textit{Third}, the data generation process can be done locally on the user side, which eliminates the need for a trusted server (that can be attacked and compromised) for protecting the private data from users. \textit{Finally}, we can prove that any machine learning model trained over the generated data is also differentially private w.r.t. the private data. 
 
To achieve this, we build two distinct differentially private generative models. First, we propose a \textbf{d}ifferentially \textbf{p}rivate \textbf{au}toencoder-based \textbf{g}enerative \textbf{m}odel  (\dpautogm). \dpautogm works for the scenario when private data is sensitive (i.e., not suitable for releasing to public) while sharing it with other parties will facilitate data analytics. As motivation, consider a hospital not allowed to release its private medical data to public for use, but wants to share the data with universities for, say, data-driven disease diagnosis studies~\cite{hett2018graph,schulam2015framework}. Under this scenario, instead of publishing the medical data directly, the hospital could locally use the  private medical data to train an autoencoder in a differentially private way~\cite{abadi2016deep} and then publish it. Any university interested in researching disease diagnosis independently feeds into the autoencoder their own small amounts of sanitized/public medical data for generating new data for machine learning tasks. Here, the sanitized/public data often refers to the publicly available data, such as~\cite{azimo,disease}. Ultimately, the private medical data owned by the hospital are successfully synthesized with public data owned by each university in a differentially private manner, so that the privacy of the private medical data is preserved and the utility of the data-driven disease study is retained. 
Another motivating example is two companies that want to collaborate on a data intelligence task. A data-rich company X may wish to aid company~Y in developing a model that helps maximize revenue, but is unwilling or legally unable to share its data with Y directly due to their sensitive nature. Again, company X can train a differentially private autoencoder (i.e., \dpautogm), on its large data set and share it with company Y. Then company Y could use its own, smaller, dataset along with the autoencoder to train a model that synthesizes information from both datasets. 

The key advantage of the \dpautogm approach is that the representation-learning task performed on the private data significantly boosts the accuracy of the machine learning task, as compared to using the public\footnote{
For simplicity, we refer to the dataset that is passed through the autoencoder as public. In the second motivating example, if company Y only uses the trained model internally, both X's and Y's datasets remain private, but from an analysis perspective, we focus on the potential privacy leaks of X's data through either the shared autoencoder or the final trained model.} data alone, in cases where the public data has too few samples to successfully train a deep learning model~\cite{lecun2015deep}. We demonstrate this using extensive experiments on four datasets (i.e., \mnist, \adult, \hospital, and \malware), showing that \dpautogm can achieve high data utility even under a small privacy budget (i.e., $\epsilon<1$) for private data.

Second, we  propose a \textbf{d}ifferentially \textbf{p}rivate \textbf{va}riational auto\textbf{e}ncoder-based \textbf{g}enerative \textbf{m}odel (\dpvae). Compared with the ordinary autoencoder, the VAE~\cite{kingma2013auto} has an extra sampling layer which can sample data from a Gaussian distribution. Using this feature, \dpvae is capable of generating an arbitrary amount of data by feeding Gaussian noise to the model. Similar to \dpautogm, the proposed \dpvae is trained on the private data in a differentially private way~\cite{abadi2016deep} and is then released to public for use. 
Although imposing Gaussian noise on the sampling layer is useful in generating new data (i.e., capable of generating infinite data), we identify that the VAE does not perform stably in generating high-quality data points. Thus, in our paper, we only evaluate \dpvae on the image dataset \mnist. We show that the data generated from \dpvae can successfully retain high utility and preserve data privacy. Under the setting of $\epsilon=8$ and $\delta=10^{-3}$, the prediction accuracy of \dpvae is more than $97\%$ on \mnist.

To further demonstrate the robustness of our two proposed models, we evaluate both \dpautogm and \dpvae with three existing attacks---model inversion attack~\cite{usenix2014,ccs15}, membership inference attack~\cite{shokri2016membership}, and GAN-based attack against collaborative deep learning~\cite{hitaj2017deep}. The results show that \dpautogm can effectively mitigate all of the aforementioned attacks and \dpvae is robust against the membership inference attack. As both \dpautogm and \dpvae satisfy differential privacy, while only \dpautogm is robust to the model inversion and GAN-based attacks, we conjecture that the key to defend against these two attacks is not due to differential privacy but the perturbation of training data.

Finally, we integrate our proposed generative models with two real-world applications, which are threatened by the aforementioned attacks. The first application is \emph{machine learning as a service} (MLaaS). Traditionally, users need to upload all of their data to the MLaaS (such as Amazon Machine Learning~\cite{amazon}) to train a model, due to the lack of computational resources on the user side. However, if these platforms are compromised, all of the users' data will be leaked. Thus, we propose to integrate \dpautogm and \dpvae with this application, so that even if the platforms are compromised, the privacy of users' data can still be protected. We empirically show that after being integrated with \dpautogm and \dpvae, this application still maintains high utility. The second application is \emph{federated learning}~\cite{mcmahan2017federated}, which has been recently shown to be vulnerable to GAN-based attacks~\cite{hitaj2017deep}. As \dpautogm is more effective in defending against this attack, we try to combine \dpautogm with this application. We show that for federated learning, even under small privacy budgets ($\epsilon=1$, $\delta=10^{-5}$), \dpautogm only decreases original utility by~$5\%$. 

The contributions of this paper are as follows:
\begin{itemize}
\item We propose two differentially private data generative models \dpautogm and \dpvae, which can provide differential privacy guarantees for the generated data, and retain high data utility for various machine learning tasks.
In addition, we compare the learning efficiency of the generated data with state-of-the-art private training methods. We show that the utility of \dpautogm outperforms \emph{Deep Learning with Differential Privacy} (\dpdl)~\cite{abadi2016deep} and \emph{Scalable Private Learning with PATE} (\spate)~\cite{papernot2018scalable} under any given privacy budget. We also show that \dpvae can achieve comparable learning efficiency in comparison with \dpdl.
\item We empirically evaluate and demonstrate that the proposed model \dpautogm is robust against existing privacy attacks---model inversion attack, membership inference attack, and GAN-based attack against collaborative deep learning;  \dpvae is robust against the membership inference attack. We conjecture that the key to defend against model inversion and GAN-based attacks is to distort the training data while differential privacy is targeted to protect membership privacy.
\item We integrate the proposed generative models with \emph{machine learning as a service} and \emph{federated learning} to protect data privacy. We show that such integration can retain high utility for these real-world applications, which are currently threatened by privacy attacks. 
\end{itemize}

To the best of our knowledge, this is the first paper to build and systematically examine differentially private data generative models that can defend against  contemporary privacy attacks on learning systems.

\section{Background}
In this section, we introduce some details about privacy attacks, differential privacy, and data generative models.

\subsection{Privacy Attacks on Learning Systems}
\noindent \textbf{Model Inversion Attack.}
This attack was first introduced by Fredrikson et al.~\cite{usenix2014} and further developed in~\cite{ccs15}. The goal of this attack is to recover sensitive attributes within original training data. For example, an attacker can infer the genome type of patients from medical records data or recover distinguishable photos by attacking a facial recognition API. Such a vulnerability mainly results from the rich information captured by the machine learning models, which can be leveraged by the attacker to recover original training data by constructing data records with high confidence. In this paper, we mainly focus on a strong adversarial scenario where attackers have white-box access to the model so as to evaluate the robustness of our proposed differentially private mechanisms. In this context, an attacker aims to reconstruct data used in the training phase by minimizing the difference between hypothesized and obtained confidence vectors from the machine learning models.

\noindent \textbf{Membership Inference Attack.}
Shokri and Shmatikov~\cite{shokri2016membership} proposed the membership inference attack to determine whether a specific data record is within the training set. This attack also takes advantage of rich information recorded in machine learning models. An attacker first generates data with similar distribution as the original data by querying machine learning models and then uses the generated data to train local models (termed shadow models in \cite{shokri2016membership}) to mimic the behavior of the original models. Finally, the attacker can apply the data provided by the local models to training a classifier and determine whether a given record belongs to the original training dataset.

\noindent \textbf{GAN-based Attack against Collaborative Deep Learning.}
Hitaj et al.~\cite{hitaj2017deep} proposed a GAN-based attack targeting  differentially private collaborative deep learning~\cite{shokri2015privacy}. They showed that an attacker may use GANs to generate instances which well approximate data from other parties in a collaborative setting. The adversarial generator is improved based on the information returned from the trusted entity, and eventually achieves high attack success rate in the collaborative scenario even when differential privacy is guaranteed for each party.

\subsection{Differential Privacy}
Differential privacy provides strong privacy guarantees for data privacy analysis~\cite{dwork2014algorithmic}. It ensures that attackers cannot infer sensitive information about input datasets merely based on the algorithm outputs. The formal definition is as follows.
 \begin{defn}\label{def:dp1}
 \label{dpdef}
A randomized algorithm $\mathcal{A}: \mathcal{D} \to \mathcal{R}$ with domain~$\mathcal{D}$ and range $\mathcal{R}$, is ($\epsilon, \delta$)-differentially private if for any two adjacent training datasets $d, d^{\prime}  \subseteq \mathcal{D}$, which differ by at most one training point, and any subset of outputs $S \subseteq \mathcal{R}$, it satisfies that:
\begin{equation*}
            \Pr[\mathcal{A}(d) \in S] \leq \mathrm{e}^{\epsilon} \Pr[\mathcal{A}(d^{\prime}) \in S] + \delta.
\end{equation*}
 \end{defn}

The $\epsilon$ parameter is often called a privacy budget: smaller budgets yield stronger privacy guarantees. The second parameter $\delta$ is a failure rate for which it is tolerated that the privacy bound defined by $\epsilon$ does not hold.

\noindent \textbf{Deep Learning with Differential Privacy (\dpdl)~\cite{abadi2016deep}.} \dpdl achieves DP by injecting random noise in stochastic gradient descent (SGD) algorithm. At each step of SGD, \dpdl computes the gradient for a random subset of training points, followed by clipping, averaging out each gradient, and adding noise in order to protect privacy. \dpdl provides a differentially private training algorithm with tight DP guarantees based on moments accountant analysis~\cite{abadi2016deep}.

\subsection{Data Generative Models}

\noindent \textbf{Autoencoder.}
An autoencoder is a widely used unsupervised learning model in many scenarios, such as natural language processing~\cite{deng2010binary} and image recognition~\cite{masci2011stacked}. Its goal is to learn a representation of data, typically for the purpose of dimensionality reduction~\cite{dp_book, denoise_auto, stack_auto}. It simultaneously trains an encoder, which transforms a high-dimenstional data point to a low-dimensional representation, and a decoder, which reconstructs a high-dimensional data point from the representation, while trying to   minimize the $2$-norm distance $l_{2}$ between the original and reconstructed data. Through this process, the autoencoder is able to discard those irrelevant features and enhance the performance of machine learning models when facing high-dimensional input data.

\noindent \textbf{Variational Autoencoder (VAE).} Resembling the autoencoder, an variational autoencoder also comprises two parts: the encoder and the decoder~\cite{kingma2013auto,rezende2014stochastic} with a latent variable $z$ sampled from a prior distribution $p(z)=p_{noise}$. Different from the autoencoder of which the encoder only tries to reduce the data into lower dimensions, the encoder inside VAE tries to encode the input data into a Gaussian probability density domain~\cite{kingma2013auto}. Mathematically, the encoder approximates $q(z|x)$, which is also a neural network (encoder), with input $z$ conditioned on the data $x$. Then, a representation of the data will be sampled based on the output from the encoder. Finally, the decoder tries to reconstruct a data point based on sampled noise, which approximates the posterior $p(x|z)$. The two neural networks, encoder and decoder, are trained to maximize a lower bound of the log-likelihood of the data $\log p(x)$:
\begin{equation*}
\mathbb{E}_{q(z|x)}[\log p(x|z)]-{\mathrm{KL}}(q(z|x)||p(z)),
\end{equation*}
where ${\mathrm{KL}}$ is the Kullback-Leibler divergence~\cite{Cover1959Information}.

Sampling from the VAE is achieved by sampling from the (typically Gaussian) prior~$p(z)$ and passing the samples through the decoder network.

\section{Differentially Private Data Generative Models}

\subsection{Problem Statement}\label{statement}
Let $\mathcal{X}$ be the set of training data containing sensitive information, and we will denote it as private data similarly with~\cite{papernot2016semi}. We denote~$\mathcal{M}$ as a data generative model which is trained on the private data, and is able to generate new data $\mathcal{X}^{\prime}$ for later training usage, as shown in Figure~\ref{architecture}. 
To protect privacy of the private data, the goal of the generative model is to prevent an attacker from recovering~$\mathcal{X}$, or inferring sensitive information from $\mathcal{X}$ based on~$\mathcal{X}^{\prime}$. Formally, we give the definition of the differentially private generative model as below.

\begin{defn}
\label{def:dp2}
A generative model $\mathcal{M}: \mathcal{D} \to \mathcal{Z}$ with domain~$\mathcal{D}$ and range $\mathcal{Z}$,
is $(\epsilon, \delta)$-differentially private, if for any adjacent private datasets $\mathcal{X}, \hat{\mathcal{X}} \subseteq \mathcal{D}$ which only differ by one entry, and any subset of output space $S \subseteq \mathcal{Z}$, it satisfies that:  
$$\Pr [\mathcal{M}(\mathcal{X}) \in S] \leq e^\epsilon \Pr[\mathcal{M}(\hat{\mathcal{X}}) \in S]+\delta.$$
\end{defn}
The goal of the proposed differentially private generative model is to generate data with high utility while protecting sensitive information within the data. Current research shows that even algorithms proved to be differentially private can also leak private information in the face of certain carefully crafted attacks on different levels. Therefore, in this paper, we will also analyze several existing attacks to show that the proposed differentially private generative models can defend against the state-of-the-art attacks.

\begin{figure*}[htbp]
\centering
\includegraphics[width=0.9\textwidth]{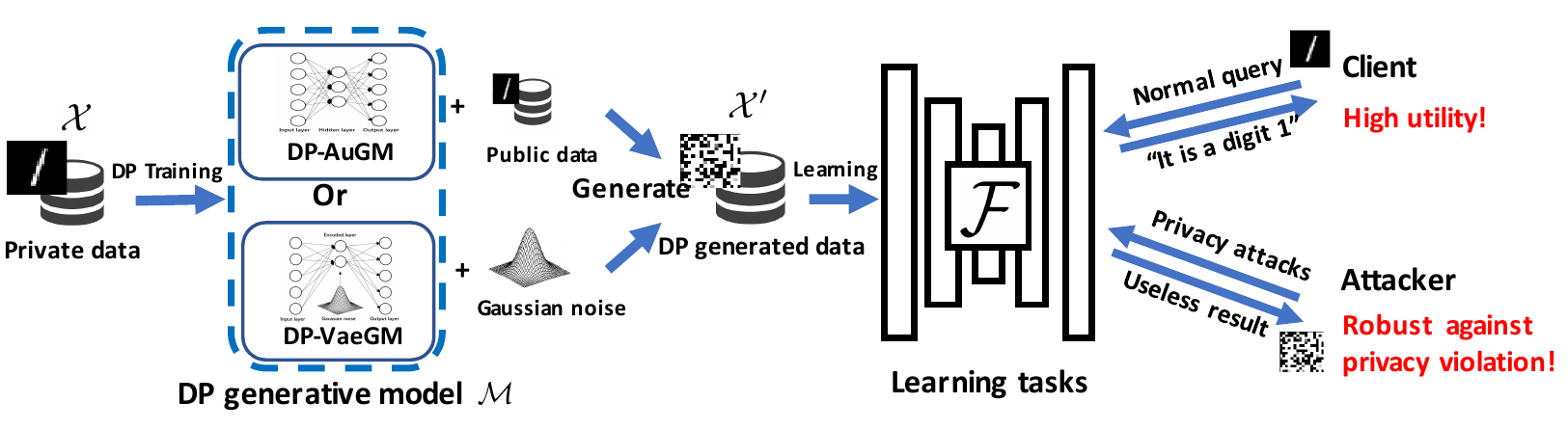}
\caption{Overview of proposed differentially private data generative models. Sensitive private training data $\mathcal{X}$ is fed into the generative model $\mathcal{M}$ to generate private surrogate dataset $\mathcal{X}^{\prime}$. After publishing $\mathcal{X}^{\prime}$, different learning models can be trained on $\mathcal{X}^{\prime}$ to protect privacy of $\mathcal{X}$ while achieving high learning accuracy (data utility).} 
\label{architecture}
\end{figure*}

\subsection{Approach Overview}
To protect private data privacy, we propose to use the private data to train a differentially private generative model and use this generative model to generate new synthetic data for further learning tasks, which can both protect privacy of original data and retain high data utility. 
As the newly generated data is differentially private w.r.t. the private data, it will be hard for attackers to recover or synthesize the private data,  or infer other information about the private data in learning tasks. Specifically, we choose an autoencoder and a variational autoencoder (VAE) as our two generative models. The overview of our proposed differentially private data generative models is shown in Figure~\ref{architecture}. First, the private data is used to train the generative model with differential privacy, which is either an autoencoder (\dpautogm) or a variational autoencoder (\dpvae) based model. Then the generated data from the trained differentially private generative model is published and sent to targeted learning tasks. It should be noted that \dpautogm requires the users to hold a small amount of data (denoted as public data in the figure) to generate new data while \dpvae is able to directly generate an arbitrary number of new data points by feeding Gaussian noise into the model. The goal of our design is to ensure that the learning accuracy on the generated data is high for ordinary users (high data utility), while the attackers cannot obtain sensitive information from the private data.

\subsection{Privacy and Utility Metrics}
\label{metric}

Here we will briefly introduce privacy and data utility metrics used throughout the paper.

\noindent \textbf{Privacy Metric.}
We refer to the privacy budget ($\epsilon$, $\delta$) as the privacy metric during evaluation.
We then evaluate how robust the proposed generative models are against three state-of-the-art attacks---model inversion attack~\cite{fredrikson2015model}, membership inference attack~\cite{shokri2016membership}, and GAN-based attack against collaborative deep learning~\cite{hitaj2017deep}. Specifically, to quantitatively evaluate how our models deal with the membership inference attack, we use the metric \textit{privacy loss} as defined in~\cite{pyrgelis2017knock}.

\noindent \textit{\textbf{Privacy Loss (PL).}}
Within membership inference attack, we measure the privacy loss as the inference precision increment over random guessing baseline (e.g., 0.5), where the adversary's attack precision rate $P$ is defined as the fraction of records that are correctly inferred as members of the training set among all the positive predictions. We define privacy loss $PL$ as follows:
\begin{equation*}
    PL=
\begin{cases}
   \frac{P-0.5}{0.5},   &  \text{if  $P> 0.5$} \\
   $0$,
        &  \text{otherwise}
\end{cases}
\end{equation*}

\noindent \textbf{Utility Metric.}
We use the prediction accuracy to measure utility for different models. Considering the goal of machine learning is to build an effective prediction model, it is natural to evaluate how our proposed model performs in terms of prediction accuracy. To be specific, we will evaluate the prediction model which is trained on the generated data from the differentially private generative model.

\subsection{DP autoencoder-based Generative Model (\dpautogm)}
Here we introduce how to apply the differentially private autoencoder-based generative model (\dpautogm) to protect  privacy of the private data while retaining high utility for the generated data.

For \dpautogm, we first train an autoencoder with our private data using a differentially private training algorithm. Then, we publish the encoder and drop the decoder. New data will be generated (encoded) by feeding the users' own data (i.e., public data) into the encoder. These newly generated data can be used to train the targeted learning systems in the future with privacy guarantees for the private data. In this way, the generated data could synthesize the information from both private data and public data which enables high learning efficiency, and provide privacy guarantees for private data at the same time. As we will show in the evaluation section, the user only needs a small amount of data to achieve good learning efficiency and we also compare the learning efficiency when the user only uses his own data to do the training. During inference time, the encoder will also be used to encode the test data for model predictions. Since the encoder is differentially private w.r.t. private data, publishing the encoder does not compromise privacy. 

\dpautogm proceeds as below:
\begin{itemize}
\item First, it is trained with private data using a differentially private algorithm.  
\item Second, it generates new differentially private data by feeding the public data to the encoder.
\item Third, it uses the generated data to train any machine learning model.
\end{itemize}

\noindent \textbf{DP Analysis for \dpautogm.} In this paper, we adopt the training algorithm developed by Abadi et al.~\cite{abadi2016deep} to achieve differential privacy. Based on the moments accountant technique applied in~\cite{abadi2016deep}, we obtain that the training algorithm is $(\mathcal{O}(q\epsilon \sqrt{T}),\delta)$-differentially private. Here $T$ is the number of training steps, $q$ is the sampling probability, and ($\epsilon$, $\delta$) denotes the privacy budget~\cite{abadi2016deep}. Further, by applying the post-processing property of differential privacy~\cite{dwork2014algorithmic}, we can guarantee that the generated data is also differentially private w.r.t. the private data and shares the same privacy bound with the training algorithm. In addition, we will also prove that any machine learning model which is trained on the generated data from \dpautogm, is also differentially private w.r.t. the private data and shares the same privacy bound. This also shows the benefit of training a differentially private generative model: we only need to train one DP generative model and all the machine learning models which are trained over the generated data will be differentially private w.r.t. the private data.
\begin{Thm}\label{Thm_dpaugm}
Let $\mathcal{M}$ denote the differentially private generative model and $\mathcal{X}$ be the private data. Any machine learning model trained over the generated data $\mathcal{M}(\mathcal{X})$, is also differentially private w.r.t. the private data $\mathcal{X}$.  
\end{Thm}
\begin{proof}
We denote $\mathcal{F}(\mathcal{X})$ the machine learning model trained on $\mathcal{X}$, and $\mathcal{F}(\mathcal{M}(\mathcal{X}))$ the learning model trained over the generated data. Then the proof is immediate by directly applying the post-processing property of differential privacy~\cite{dwork2014algorithmic}. 
\end{proof}

\subsection{DP Variational autoencoder-based Generative Model (\dpvae)}
In this section, we will propose \dpvae which could generate an arbitrary number of data points for usage.

\dpvae proceeds as below:
\begin{itemize}[wide=1pt,leftmargin=10pt]
\item First, it initializes with $n$ variational autoencoders (VAEs), where $n$ is the number of the classes for the specific data. Each model $\mathcal{M}_{i}$ is responsible for generating the data of a specific class $1 \leq i \leq n$. We empirically observe that training $n$ generative models results in higher utility than training a single model; we expect this is because a single model would need to capture the class label latent variables following a Gaussian distribution. Using $n$ separate models can also be used to generate a balanced dataset even if the original data are imbalanced.

\item Second, it uses a differentially private training algorithm (such as \dpdl) to train each generative model $\mathcal{M}_{i}$. 

\item Third, it samples data from Gaussian distribution $\mathcal{N}(0,1)$ for the sampling layer of each variational autoencoder. It returns the entire generated data $\mathcal{X}^{\prime}$ by taking the union of generated data from each generative model $\mathcal{M}_{i}$.

\end{itemize}

\noindent \textbf{DP Analysis for \dpvae.} We have adopted the algorithm developed by Abadi et al.~\cite{abadi2016deep} to train each VAE. Thus each training algorithm is $(\mathcal{O}(q\epsilon \sqrt{T}),\delta)$-differentially private. Next we prove that each variational autoencoder (VAE) is a differentially private generate model (see Theorem~\ref{Thm_VAE}) and the entire \dpvae 
is also $(\mathcal{O}(q\epsilon \sqrt{T}),\delta)$-differentially private (see Theorem~\ref{dp_vae}). Formally, to show proofs, we let $\mathcal{X}$ be the private data, $\Theta$ be model parameters, and $\mathcal{X}^{\prime}$ be the generated data (the output of a single VAE). 

\begin{Thm}\label{Thm_VAE}
Let $\mathcal{T}:\mathcal{X} \rightarrow \Theta$ be a VAE training algorithm that is $(\epsilon,\delta)$-differentially private based on~\cite{abadi2016deep}. Let  $f:\Theta \rightarrow \mathcal{X}^{\prime}$ be a mapping that maps model parameters to output, with Gaussian noise generated from a sampling layer of VAE as input. Then $f \circ \mathcal{T}: \mathcal{X} \rightarrow \mathcal{X}^{\prime}$ is $(\epsilon,\delta)$-differentially private.
\end{Thm}

\begin{proof}
The proof is immediate by applying the post processing property  of differential privacy~\cite{dwork2014algorithmic}.
\end{proof}
 
\begin{Thm}\label{dp_vae}
Let a generative model (VAE) of class $i$ $\mathcal{M}_{i}:\mathcal{X}_{i} \rightarrow \mathcal{X}^{\prime}_{i}$ be $(\epsilon,\delta)$-differentially private. Then if $\mathcal{G}_n:\mathcal{X} \rightarrow \Pi_{i=1}^{n} \mathcal{X}^{\prime}_{i}$ is defined to be $\mathcal{G}_n=  \bigcup_{i=1}^{n} \mathcal{M}_{i}$, $\mathcal{G}_n$ is $(\epsilon,\delta)$-differentially private, for any integer~$n$.
\end{Thm}
\begin{proof}
Given two adjacent datasets $\mathcal{X}_{1}$ and $\mathcal{X}_{2}=\mathcal{X}_{1}\bigcup\{b\}$, without loss of generalization, assume $b$ belongs to class~$k\ (1 \leq k \leq n)$. Fix any subset of events $S \subseteq \Pi_{i=1}^{n} \mathcal{X}^{\prime}_{i}$. Since the $n$ generative models are pairwise independent, we obtain $\Pr[\mathcal{G}_n(\mathcal{X}_1)\in S]=\Pi_{i=1}^{n}\Pr[\mathcal{M}_i(x_{i}^{1})\in S]$, where $x_{i}^{1} \subseteq \mathcal{X}_{1}=\bigcup_{i=1}^{n}x_{i}^{1}$ denotes the training data of $\mathcal{X}_i$ for the $i$th generative model. Similarly, $\Pr[\mathcal{G}_n(\mathcal{X}_2)\in S]=\Pi_{i=1}^{n}\Pr[\mathcal{M}_i(x_{i}^{2})\in S]$. Since $\mathcal{X}_1$ and $\mathcal{X}_2$ only differ in $b$, we have $x_{i}^{1}=x_{i}^{2}$ and $\Pr[\mathcal{M}_i(x_{i}^{1})\in S]=\Pr[\mathcal{M}_i(x_{i}^{2})\in S]$, for any $i \neq k$. Since $\mathcal{M}_k$ is $(\epsilon,\delta)$-differentially private, then we have $\Pr[\mathcal{M}_k(x_{k}^{1})\in S]\leq e^{\epsilon}\Pr[\mathcal{M}_k(x_{k}^{2})\in S]+\delta$. Therefore, we obtain $\Pr[\mathcal{G}_n(\mathcal{X}_{1})\in S]=\Pi_{i=1}^{n}\Pr[\mathcal{M}_i(x_{i}^{1})\in S]=
\Pr[\mathcal{M}_1(x_{1}^{2})\in S]\times \cdots \times \Pr[\mathcal{M}_k(x_{k}^{1})\in S]\times \cdots \times \Pr[\mathcal{M}_n(x_{n}^{2})\in S] 
\leq e^{\epsilon}\Pi_{i=1}^{n}\Pr[\mathcal{M}_i(x_{i}^{2})\in S]+\delta
=e^{\epsilon}\Pr[\mathcal{G}_n(\mathcal{X}_{2})\in S]+\delta$. The inequality derives from the fact that any probability is no greater than $1$. Hence, $\mathcal{G}_n$ is $(\epsilon,\delta)$-differentially private, for any $n$.
\end{proof}

\noindent \textbf{Remark.}
Both \dpvae and \dpautogm can realize differentially private generative models w.r.t. the private data. The main difference is that \dpautogm requires users' own data (i.e., public data) to generate new data while \dpvae can generate infinite number of data points just based on Gaussian noise. Although the feature of \dpvae is pretty good, we do notice that the generated data quality is not always stable while \dpautogm is always stable in terms of utility. More details are presented in the evaluation section.

\section{Experimental Evaluation}
\label{sec:exp}
In this section, we first describe datasets used for evaluation, followed by the empirical results of two data generative models. 
Note that all the structures of generative models and machine learning model involved in the experiments are specified in Appendix~\ref{sec: model_arc}.

\subsection{Datasets}
\noindent \textit{\textbf{\mnist.}}
\mnist~\cite{lecun1998gradient} is the benchmark dataset containing handwritten digits from 0 to 9, comprised of 60,000 training and 10,000 test examples. Each handwritten grayscale image of digits is centered in a 28$\times$28 or 32$\times$32 image. To be consistent with ~\cite{hitaj2017deep}, we choose to use the 32$\times$32 version of \mnist dataset when evaluating our generative models against the GAN-based attack.

\noindent \textit{\textbf{Adult Census Data.}}
The Adult Census Dataset~\cite{adult} includes 48,843 records with 14 sensitive attributes, including gender, education level, marital status, and occupation. This dataset is commonly used to predict whether an individual makes over 50K dollars in a year. 32,561 records serve as a training set and 16,282 records are used for testing.

\noindent \textit{\textbf{\hospital.}}
This dataset is based on the Public Use Data File released by the Texas Department of State Health Services in 2010Q1~\cite{hospital}. Within the data, there are personal sensitive information, such as gender, age, race, length of stay, and surgery procedure. We focus on the 10 most frequent main surgery procedures, and exploit part of categorical features to make inference for each patient. The resulting dataset has 186,976 records with 776 binary features. We randomly choose 36,000 instances as testing data and the rest serves as the training data.

\noindent \textit{\textbf{Malware Data.}}
To demonstrate the generality of the proposed models, we also include the Android mobile malware dataset~\cite{chen2016stormdroid} for diversity purposes. This dataset is previously used to determine whether an Android application is benign or malicious based on 142 binary features, such as user permission request. We randomly choose 3,240 instances as training data and 2,000 as testing data.

\subsection{Evaluation of \dpautogm}
\label{autogm}
In this subsection, we first show how \dpautogm performs in terms of utility under different privacy budgets on four datasets. To evaluate performance, for \mnist, we split the test data into two parts: 90\% is used as public data and the rest 10\% is used as a hold out to evaluate test performance as in~\cite{papernot2018scalable}. For \adult, \hospital, and \malware, the test data is evenly split into two halves: the first serves as public data and the second is used for evaluating test performance. All the training data is regarded as private data of which the privacy we aim to protect. Then we analyze how public data size influences \dpautogm on \mnist dataset and we also compare the learning efficacy between when only using public data for training and combining it with \dpautogm. In addition, we compare \dpautogm with some state-of-art differentially private learning methods.

\begin{figure*}[t]
  \centering
  \mbox{
   \subfloat[Accuracy of machine learning models trained on generated data by \dpautogm and pristine data (Baseline)
   under different levels of privacy on \mnist \label{accuracy_dp_mnist}]{\includegraphics[width=0.235\textwidth]{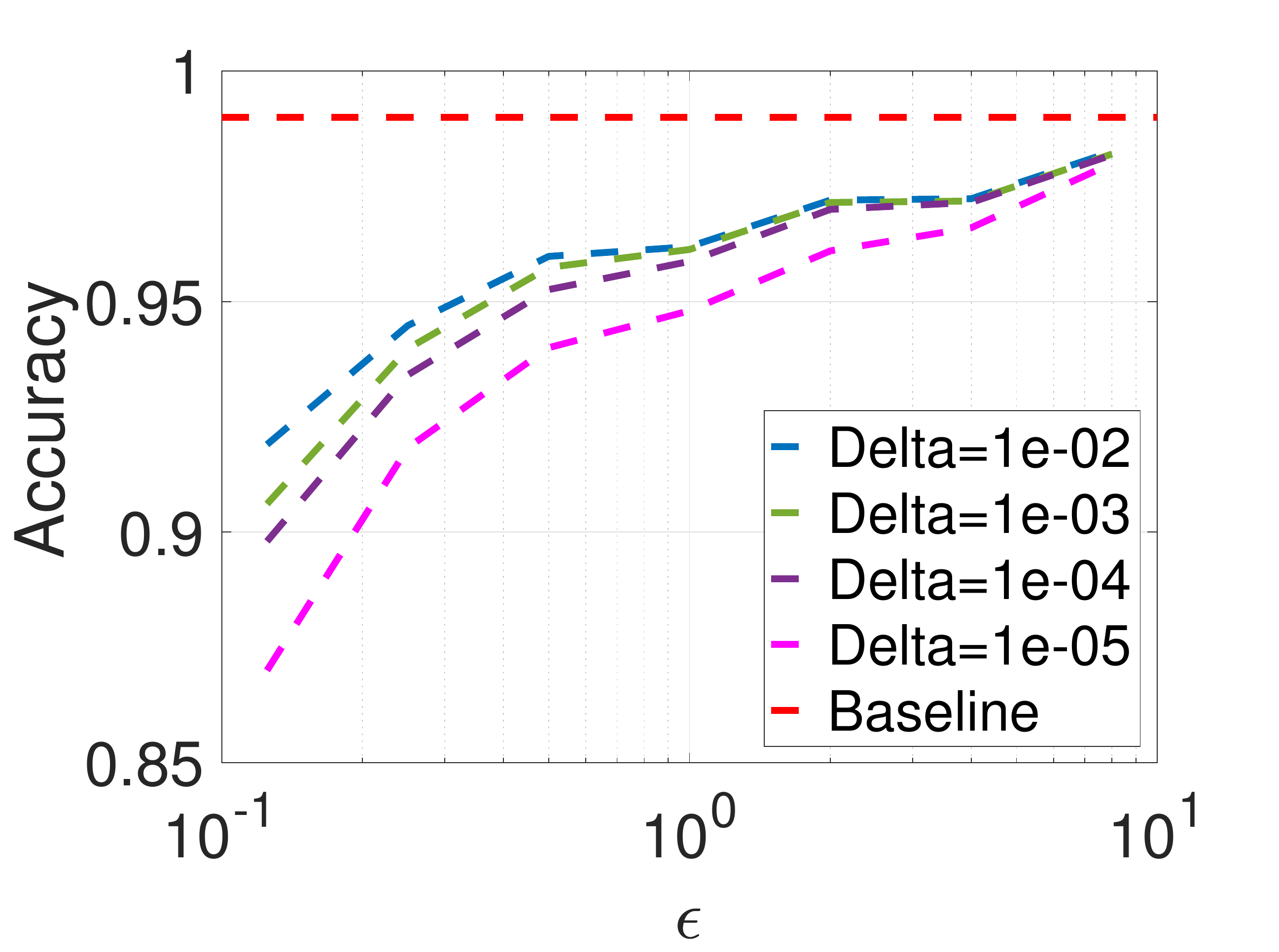}}
    \quad
      \subfloat[Accuracy of machine learning models trained on generated data by \dpautogm and pristine data (Baseline)
   under different levels privacy on \adult \label{accuracy_dp_adult}]{\includegraphics[width=0.235\textwidth]{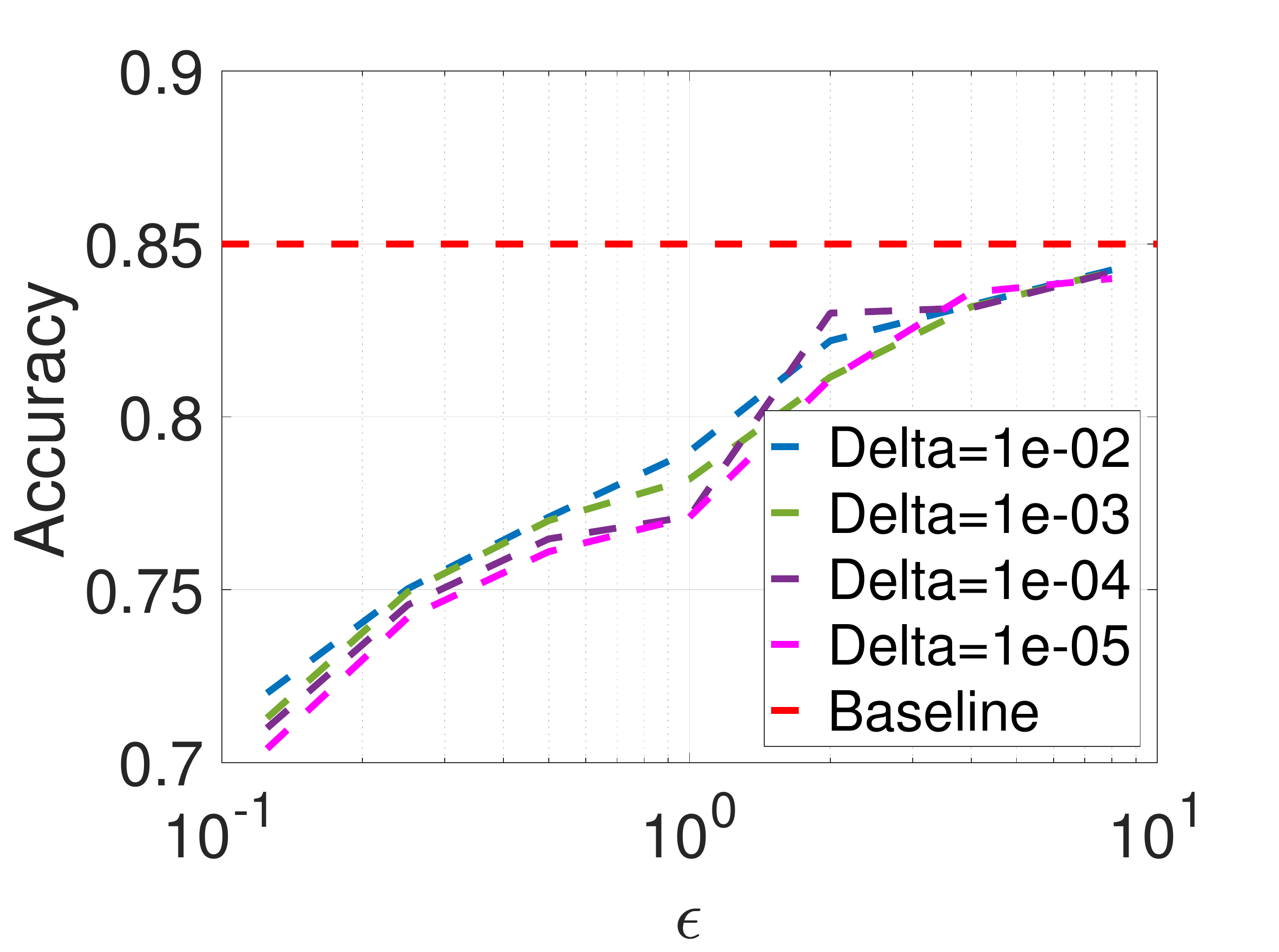}}
      \quad
     \subfloat[Accuracy of machine learning models trained on generated data by \dpautogm and pristine data (Baseline)
   under different levels privacy on \hospital \label{accuracy_dp_hospital}]{\includegraphics[width=0.235\textwidth]{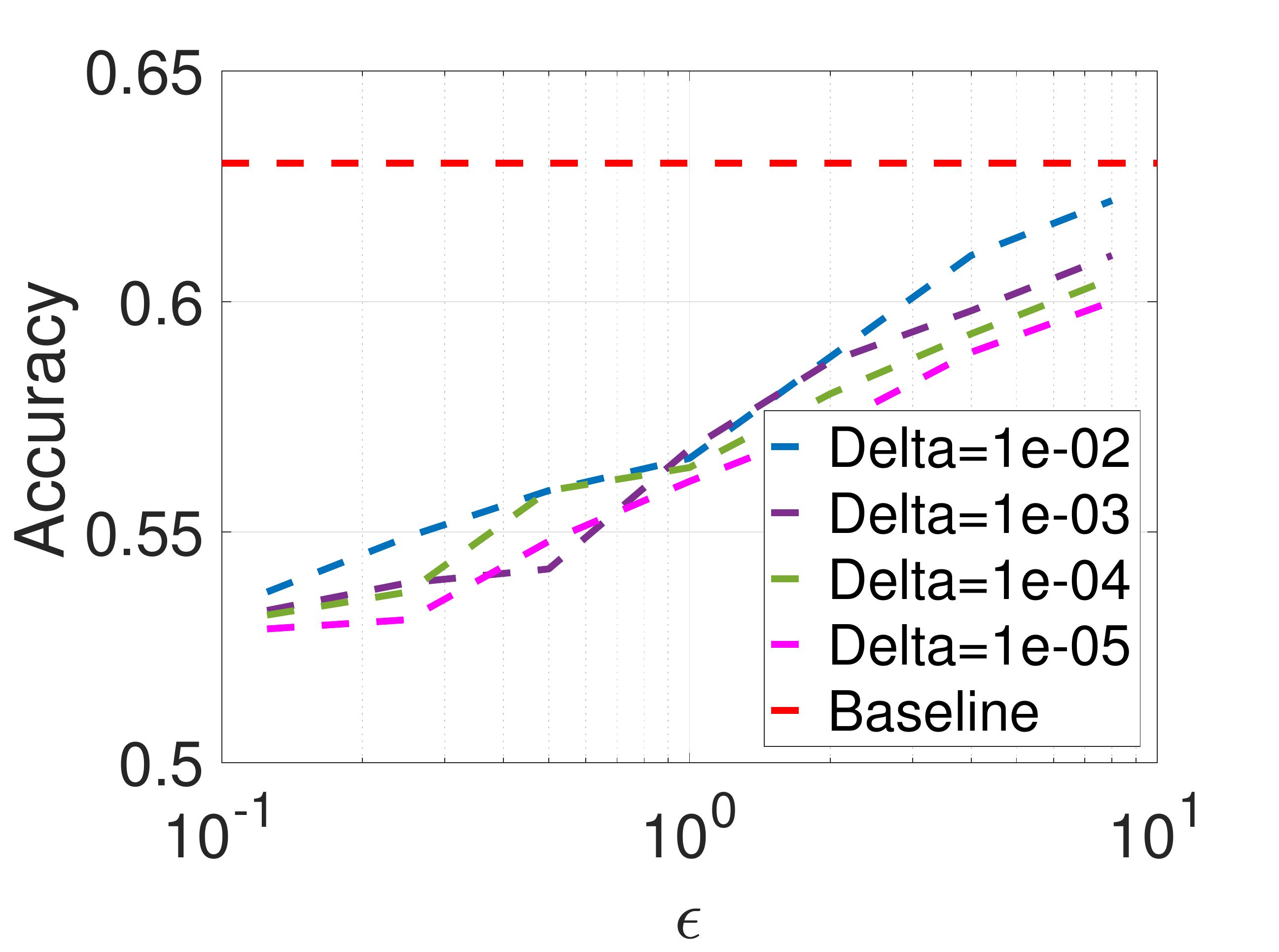}}
         \quad
        \subfloat[Accuracy of machine learning models trained on generated data by \dpautogm and pristine data (Baseline)
   under different levels privacy on \malware \label{accuracy_dp_malware}]{\includegraphics[width=0.235\textwidth]{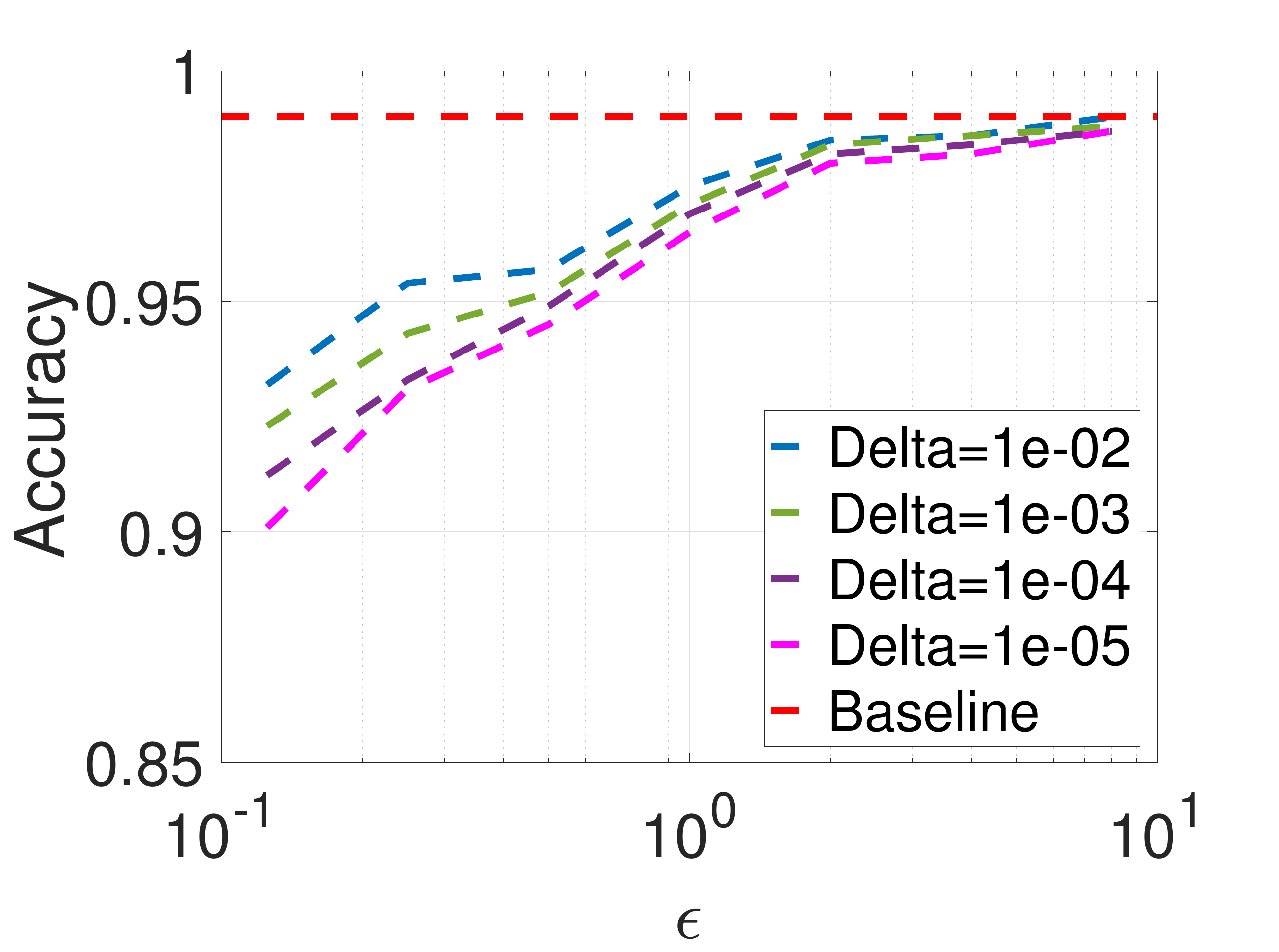}}
   }
  \caption{Evaluation of \dpautogm}
  \label{evaluation_dpaugm}
\end{figure*}

\begin{figure*}[t]
  \centering
  \mbox{
  \subfloat[Comparison between  \dpautogm and \dpdl on \mnist with $\delta=10^{-5}$ \label{dp_dpdl}]{\includegraphics[width=0.45\textwidth]{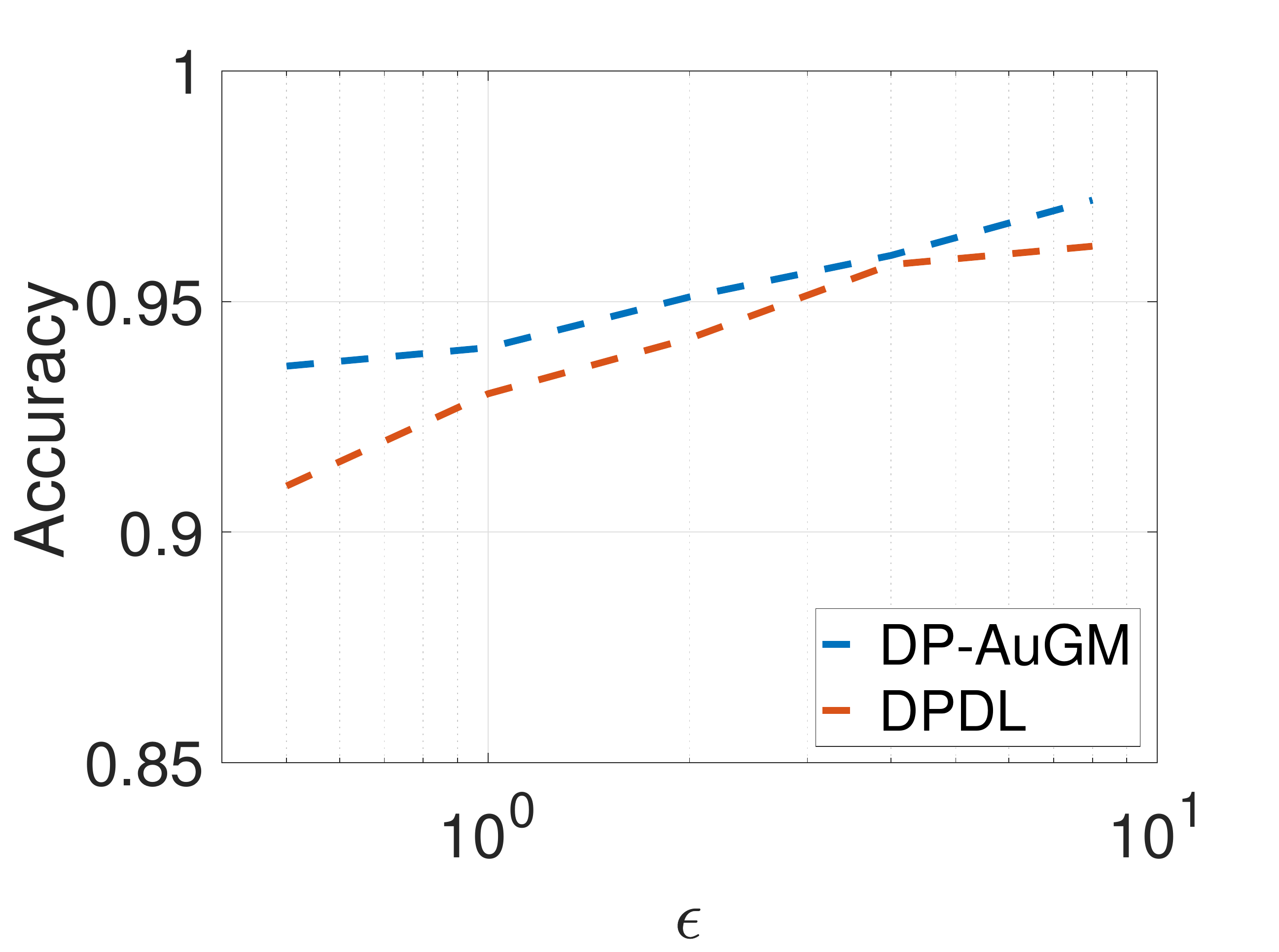}}
  \quad
   \subfloat[Comparison between \dpvae and \dpdl on \mnist with $\delta=10^{-5}$ \label{fig:vae_com}]{\includegraphics[width=0.45\textwidth]{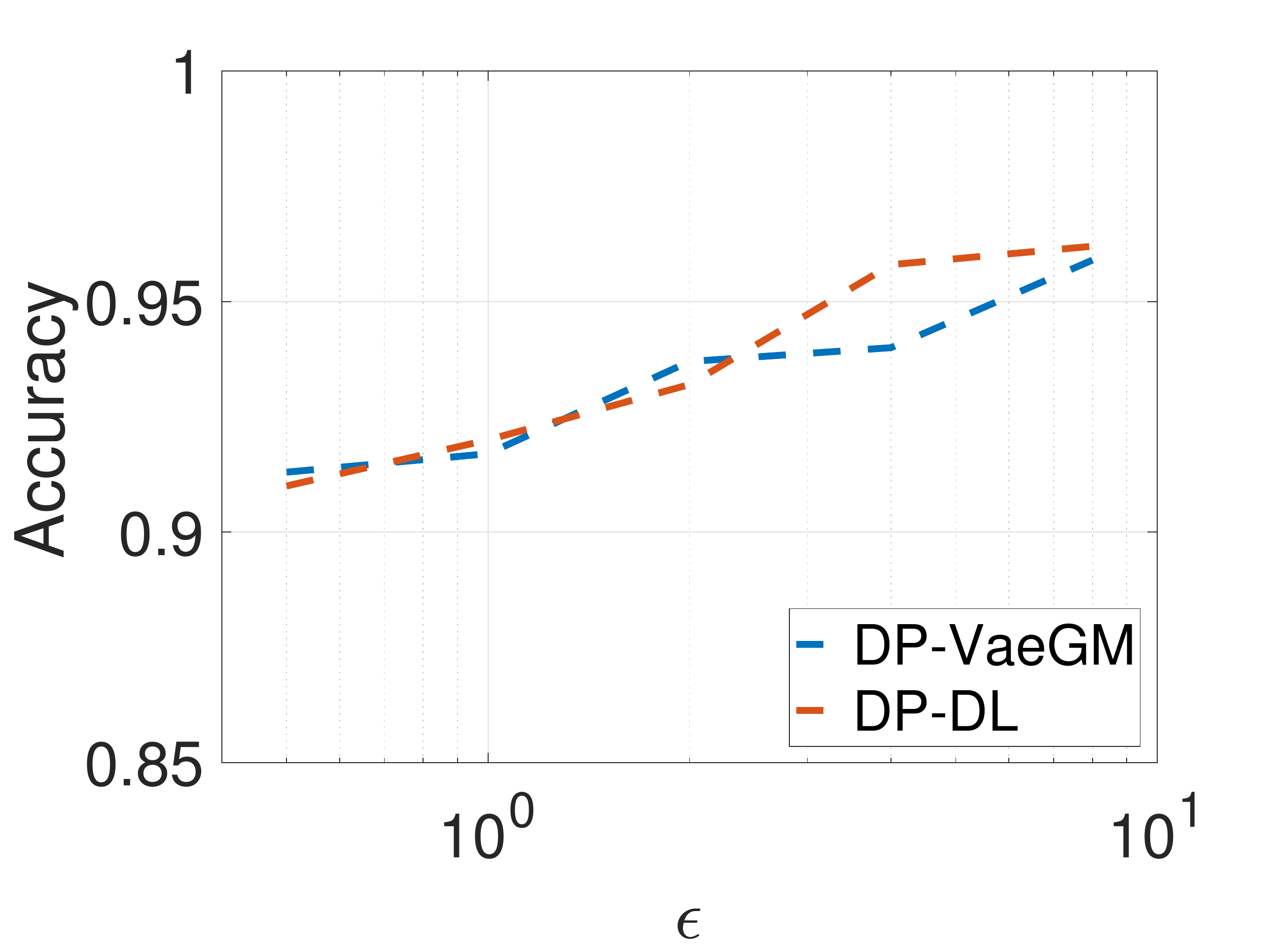}}
  }
  \caption{\dpautogm and \dpvae versus \dpdl}
  \label{evaluation_dpaugm_dpvae}
\end{figure*}

\noindent \textbf{Effect of Different Privacy Budgets.}
To evaluate the effects of privacy budgets (i.e., $\epsilon$ and $\delta$) on prediction accuracy for machine learning models,
we vary ($\epsilon$, $\delta$) to test learning efficiency (i.e., the utility metric) on different datasets. The results are shown in Figure~\ref{evaluation_dpaugm}(a)-(d). In these figures, each curve corresponds to the best accuracy achieved given fixed $\delta$, as $\epsilon$ varies between $0.2$ and $8$. In addition, we also show the baseline accuracy (i.e., without \dpautogm) on each dataset for the comparison. From Figure~\ref{evaluation_dpaugm}, we can see that the prediction accuracy decreases as the noise level increases ($\epsilon$~decreases), while we see \dpautogm can still achieve comparable utility with the baseline even when $\epsilon$ is tight (i.e., around~$1$). When $\epsilon=8$, for all the datasets, the accuracy lags behind the baseline within $3\%$. This demonstrates that data generated by \dpautogm can preserve high data utility for subsequent learning tasks.

\noindent \textbf{Efficacy of \dpautogm.}
We further examine how \dpautogm helps boost the learning efficacy. We compare the learning accuracy between only public data is used for training and by combining both \dpautogm and public data. For \dpautogm, we set the private budge $\epsilon$ and $\delta$ to be $1$ and $10^{-5}$, respectively. We do the comparisons on all the datasets and the result is presented in Table~\ref{tab:compare_dpaugm}. As we can see from Table~\ref{tab:compare_dpaugm}, after using \dpautogm, the learning accuracy increases by at least $6\%$ on all the datasets and  by $34\%$ on \malware dataset. This actually demonstrates the significance of using \dpautogm for sharing the information of private data. Since the amount of private data is huge, \dpautogm trained over the private data can better capture the inner representations of the dataset, which further boosts the following learning accuracy of machine learning models. In addition, we also examine how utility is affected when different amounts of public data is available on the dataset \mnist. We vary the public data size from 1,000 to 9,000 in steps of 1,000. The privacy budget $\epsilon$ and $\delta$ is set as $1$ and $10^{-5}$, respectively. As we can see from Figure~\ref{fig:augm_public}, the public data size affects test accuracy slightly, only within $7\%$ dropping. This suggests that private data plays a major role in generating high-utility data for learning efficacy.

\begin{figure}[t]
\centering
\includegraphics[scale=0.3]{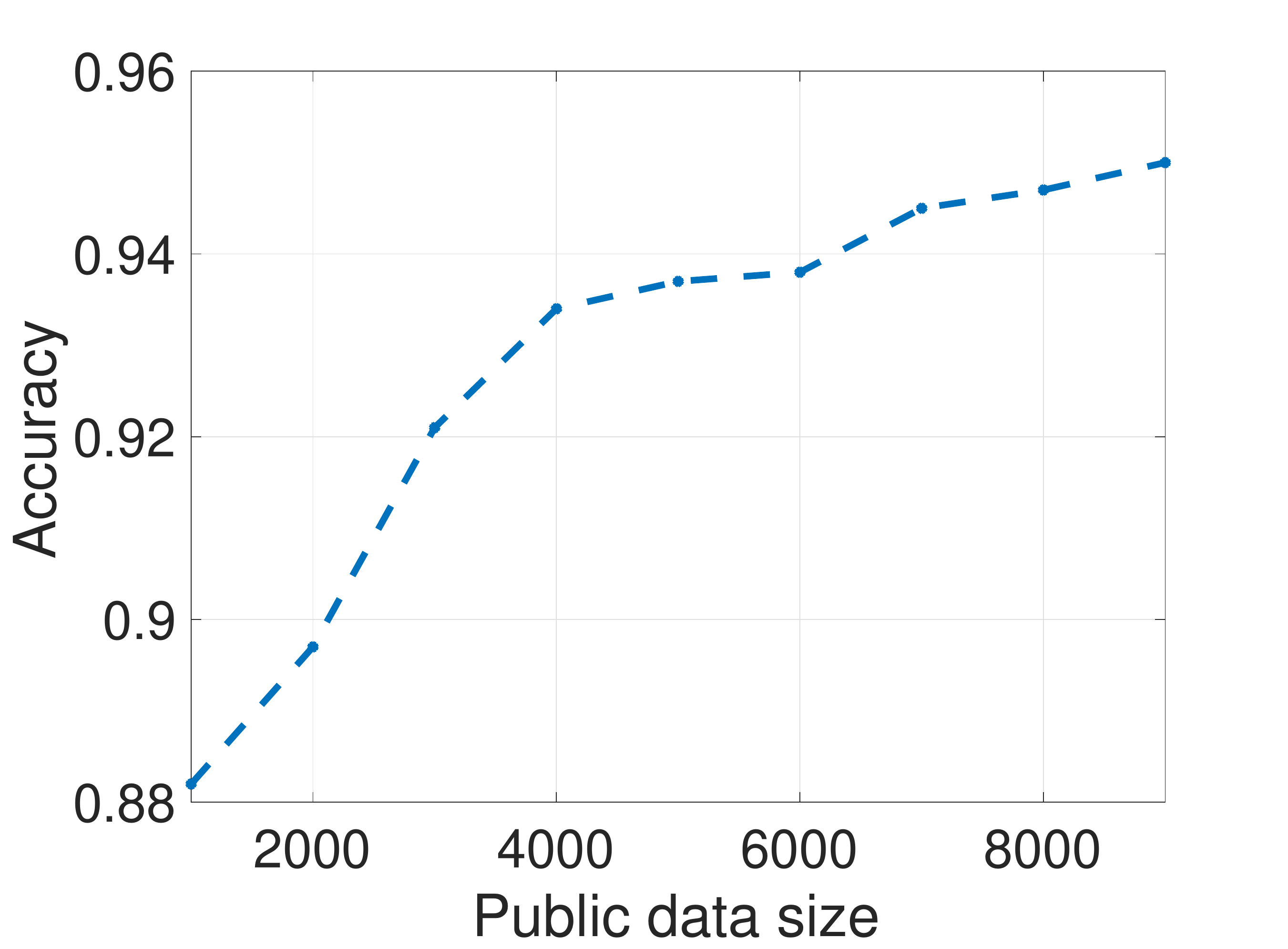}
\caption{Accuracy of \dpautogm by different sizes of public data}
\label{fig:augm_public}
\end{figure}

\begin{table}[t]
\caption{Comparisons of training accuracy between using only public data for training and using both \dpautogm and public data}
\label{tab:compare_dpaugm}
\resizebox{\linewidth}{!}{
\centering
\begin{tabular}{c|c|c}
  \toprule
  {Datasets}&{With \dpautogm}&{Without \dpautogm} \\
  \midrule
\mnist& 0.95 & 0.89\\
\adult & 0.78 &  0.64\\
\hospital & 0.56 &  0.42\\
\malware & 0.96 &  0.62\\
  \bottomrule
\end{tabular}}
\end{table}

\noindent \textbf{In Comparison with the Differentially Private Training Algorithm (\dpdl).} Although our method leverages \dpdl as the differentially private training algorithm, we show that our method better performs in training the machine learning model under the same privacy budget. For comparison, we choose the feed-forward neural network model with the architecture and \mnist dataset specified in~\cite{abadi2016deep}. 
In addition, we use 90\% of the test data as public data and the rest acts as the test data for both methods. For \dpdl, the public data simply serves as its training data. As for the privacy budget, we fix $\delta$ as $10^{-5}$ and vary $\epsilon$ from $0.5$ to $8$. The result is shown in Figure~\ref{evaluation_dpaugm_dpvae}(a). As we can see from Figure~\ref{evaluation_dpaugm_dpvae}(a), under different $\epsilon$, our method outperforms \dpdl consistently. Furthermore, \dpdl needs to be performed each time on a new model while \dpautogm only needs to be trained once. Then any model which is trained over the generated data from \dpautogm is differentially private w.r.t. the private data (i.e., with the same property of differential privacy achieved by \dpdl). Hence, we can see that \dpautogm outperforms \dpdl both in accuracy and computational efficiency.

\noindent \textbf{In Comparison with Scalable Private Learning with PATE.} Scalable Private Learning with PATE (\spate)~\cite{papernot2018scalable} is recently proposed by Papernot et al., which can also realize a differentially private training algorithm w.r.t. the private data and provides privacy protection for partial data. We try to compare \spate with \dpautogm on \mnist in terms of the utility metric. Here, the baseline denotes the scenario where no privacy protection mechanism is used. We follow~\cite{papernot2018scalable} to split the test data into two parts. One part serves as public data while the second serves as test data. We also use the same CNN machine-learning model as specified in~\cite{papernot2018scalable}. As we can see from Table~\ref{tab:compa}, \dpautogm outperforms \spate by 0.2\% in terms of prediction accuracy and only sits below the baseline by 0.5\%. Note that the reason of making a comparison at a specific pair of the privacy budget is that \spate~\cite{papernot2018scalable} only presents the result on~\mnist for a specific pair of differential privacy parameters. 
Furthermore, \dpautogm surpasses \spate in terms of computational efficiency since 250 teacher models are used in \spate while \dpautogm only needs to be trained once.

\begin{table}[t]
\caption{Comparisons between \dpautogm and \spate on \mnist}
\label{tab:compa}
\resizebox{\linewidth}{!}{
\centering
\begin{tabular}{c|c|c|c|c}
  \toprule
  {Models}&{ Privacy budget~$\epsilon$} &{Privacy budget~$\delta$} &{Accuracy}&{Baseline} \\
  \midrule
\spate~\cite{papernot2018scalable}& 1.97 & $10^{-5}$ & 0.985 &0.992\\
\dpautogm & 1.97 & $10^{-5}$ & 0.987 &0.992\\
  \bottomrule
\end{tabular}}
\end{table}

\subsection{Evaluation of \dpvae} 
In this subsection, we empirically evaluate utility performance of our proposed data generative model \dpvae. As VAE is usually used to generate high quality images, now we only evaluate \dpvae on the image dataset \mnist. 

\noindent \textbf{Effect of Different Privacy Budgets.}
We vary the privacy budget to test \dpvae on \mnist dataset. The result is shown in Figure~\ref{fig:vae_acc}, where each curve corresponds to the best accuracy given fixed $\delta$, and $\epsilon$ varies between $0.2$ and $8$. We show the baseline accuracy (i.e., without \dpvae) using the red line. From this figure, we can see that \dpvae can achieve comparable utility w.r.t. the baseline. For instance, when $\epsilon$ is greater than $1$, the accuracy is always higher than $92\%$. When $\epsilon$ is $8$ and $\delta$ is $10^{-2}$, the accuracy is over $97\%$ which is lower than the baseline by $2\%$. Thus, we can see that \dpvae has the potential to generate data with high training utility while providing privacy guarantees for private data. 

\begin{figure}[htbp]
\centering
\includegraphics[scale=0.3]{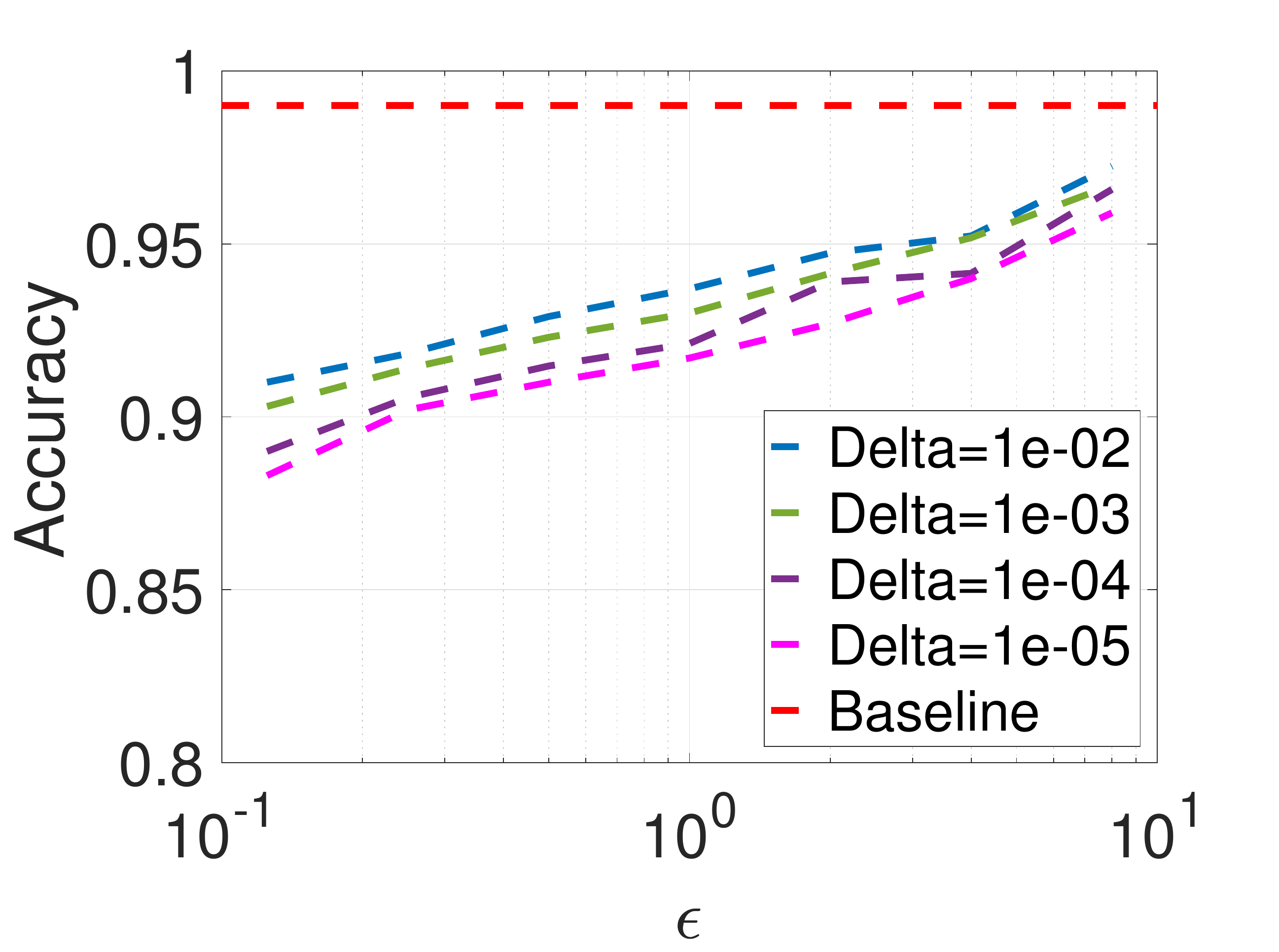}
\caption{Accuracy of \dpvae under various privacy budgets on \mnist dataset}
\label{fig:vae_acc}
\end{figure}

\noindent \textbf{In Comparison with the Differentially Private Training Algorithm (\dpdl).} We compare \dpvae with \dpdl on \mnist. As for the privacy budget, we fix $\delta$ as $10^{-5}$ and vary $\epsilon$ from $0.5$ to $8$. From Figure~\ref{evaluation_dpaugm_dpvae}(b), we can see that \dpvae achieves comparable utility with \dpdl. In addition, we want to stress that for \dpvae, once the data is generated, all machine learning models trained on the generated data will become differentially private w.r.t the private data while for \dpdl, we need to rerun the algorithm for each new model. Thus, \dpvae outperforms \dpdl in computation efficiency.

\noindent \textbf{In Comparison with Scalable Private Learning with PATE.} We also compare Scalable Private Learning with PATE (\spate)~\cite{papernot2018scalable} with \dpvae on \mnist in terms of the utility metric (i.e., prediction accuracy). The learning model applies the CNN structure as specified in~\cite{papernot2018scalable}. As \spate requires the presence of public data, we split the test data into two parts in the same way as specified by~\cite{papernot2018scalable}. Considering \dpvae does not need public data, private data is discarded for \dpvae. In addition, the privacy budget $\epsilon$ and $\delta$ is set to be $1.97$ and $10^{-5}$, respectively. From Table~\ref{tab:compa2}, we can see that \dpvae falls behind \spate by approximately 2\%. This is because that \spate trains the model using both public and private data while \dpvae is only trained with private data. 

\begin{table}[t]
\caption{Comparisons between \dpvae and \spate on \mnist}
\label{tab:compa2}
\resizebox{\linewidth}{!}{
\centering
\begin{tabular}{c|c|c|c}
  \toprule
  {Models}&{ Privacy budget~$\epsilon$} &{Privacy budget~$\delta$} &{Accuracy} \\
  \midrule
\spate~\cite{papernot2018scalable}& 1.97 & $10^{-5}$ & 0.985 \\
\dpvae & 1.97 & $10^{-5}$ & 0.968 \\
  \bottomrule
\end{tabular}}
\end{table}

\noindent \textbf{Remark.} We have empirically shown that \dpautogm and \dpvae can achieve high data utility and protect privacy of private data at the same time.

\section{Defending against Existing Attacks}
\label{sec:attacks}
To demonstrate the robustness of proposed generative models, here we evaluate the models against three state-of-the-art privacy violation attacks---model inversion attack, membership inference attack, and the GAN-based attack against collaborative deep learning. 

\subsection{Model Inversion Attack}
We choose to use the one-layer neural network to mount the model inversion attack~\cite{ccs15} over \mnist dataset setting because it is easier to check the effectiveness of the model inversion attack on image dataset. Note that Hitaj et al.~\cite{hitaj2017deep} claimed that the model inversion attack might not work on convolutional neural networks (CNN). For the original attack, we use all the training data to train the one-layer neural network and then try to recover digit 0 by exploiting the confidence values~\cite{ccs15}.  As we can see from Figure~\ref{fig:f1}, the digit 0 is almost recovered. Then, we try to evaluate how \dpautogm performs in defending against the attack. We use the generated data from \dpautogm to train the one-layer neural network. The privacy budget $\epsilon$ and $\delta$ for \dpautogm is set to be $1$ and $10^{-5}$, respectively. We then mount the same model inversion attack on the one-layer neural network. Figure~\ref{fig:f2} shows the result after deploying \dpautogm. We can clearly see that after deploying \dpautogm, nothing can be learned from the attack result as shown in Figure~\ref{fig:f2}. So we can see \dpautogm can mitigate the model inversion attack effectively. However, we find that \dpvae is not robust enough in mitigating the model inversion attack. We will discuss this in Section~\ref{discuss}.

\begin{figure}[t]
  \centering
  \mbox{
  \subfloat[Original attack]{\includegraphics[width=0.293\textwidth]{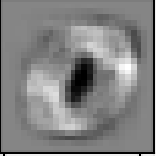}\label{fig:f1}}
   \quad
  \subfloat[Attack with \dpautogm]{\includegraphics[width=0.3\textwidth]{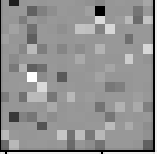}\label{fig:f2}}
  }
  \caption{The efficiency of the model inversion attack on \mnist dataset before and after deploying \dpautogm}
  \label{modelinver}
\end{figure}

\subsection{Membership Inference Attack}

\noindent We evaluate how \dpautogm and \dpvae perform in mitigating membership inference attack on \mnist using one-layer neural networks. The training set size is set to be 1,000 and the number of shadow models~\cite{shokri2016membership} is set to be 50. We have set the privacy budget $\epsilon$ and $\delta$ to be $1$ and $10^{-5}$, respectively. For this attack, we mainly consider whether this attack can predict the existence of private data in the training set. To evaluate the attack, we use the standard metric---precision, as specified in~\cite{shokri2016membership} that the fraction of the records inferred as members of the private training dataset that are indeed members. The result is shown in Figure~\ref{member_mnist}. As we can see from Figure~\ref{member_mnist}, after deploying \dpautogm, the attack precision for all the classes drops at least 10\% and for some classes, the attack precision is approaching zero, such as classes 2 and 5. Similarly for \dpvae, the attack precision drops over 20\% for all the classes. Thus, we conclude that, with \dpautogm and \dpvae, the membership inference attack can be effectively defended against. The privacy loss on \mnist is also tabulated in Table~\ref{pl_mem}. As we can see, with our proposed generative models, the privacy loss for each class can be reduced to zero. 

\begin{figure}[t]
  \centering
   \includegraphics[width=0.9\textwidth]{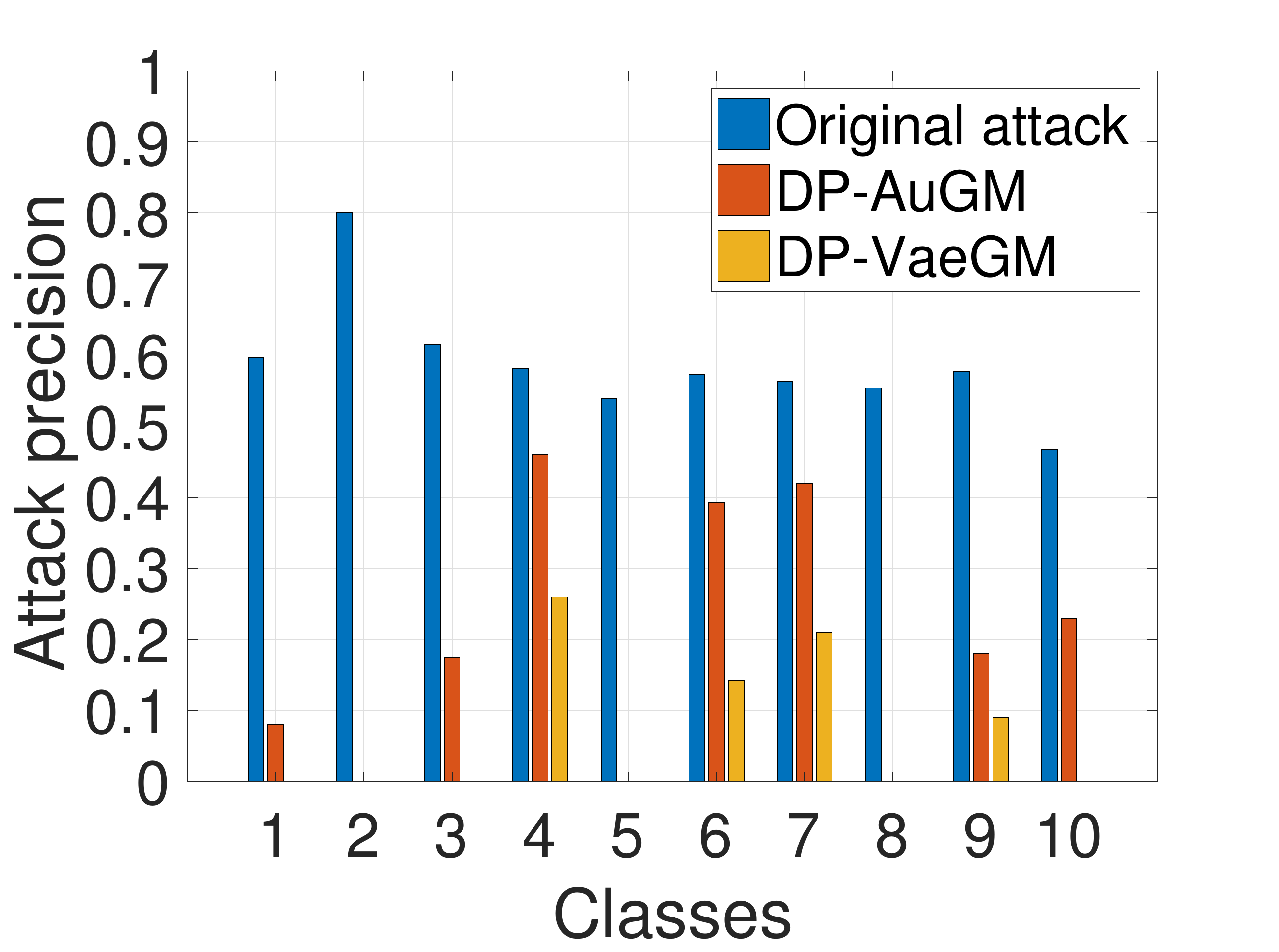}
  \caption{Evaluation of \dpautogm, and \dpvae against the membership inference attack on \mnist}
  \label{member_mnist}
\end{figure}

\begin{table}[t]\scriptsize
\caption{Privacy loss for the membership inference attack}
\label{pl_mem}
\resizebox{\linewidth}{!}{
\centering
\begin{tabular}{l|c|c|c|c|c|c|c|c|c|c}
  \toprule
     {\bf Original attack (\mnist)}&0.2&0.6&0.2&0.2&0.1&0.2&0.1&0.1&0.2&0.0\\
 \midrule
  {\bf With \dpautogm} & 0.0&0.0&0.0&0.0&0.0&0.0&0.0&0.0&0.0&0.0
  \\
  \midrule
    {\bf With \dpvae} & 0.0&0.0&0.0&0.0&0.0&0.0&0.0&0.0&0.0&0.0
    \\
   \bottomrule
\end{tabular}}
\end{table}

\subsection{GAN-based Attack against Collaborative Deep Learning}
We choose to use the \mnist dataset to analyze GAN-based attack since the simplicity of the dataset can boost the success rate for the attacker. We create two participants in this setting, where one serves as an adversary and the other serves as an honest user, as suggested in~\cite{hitaj2017deep}. We follow the same model structure as specified in~\cite{hitaj2017deep}, where the CNN is used as a discriminator and the DCGAN~\cite{radford2015unsupervised} is used as a generator. Users can apply the proposed differentially private generated data or original data to train their local models.
We show defense results for \dpautogm in Figures~\ref{federated_picture_a} and~\ref{federated_picture_b}, where Figure~\ref{federated_picture_a} represents the images obtained by adversaries without deploying generative models, while Figure~\ref{federated_picture_b} shows the obtained images which have been protected by \dpautogm. As we can see from Figure~\ref{federated_picture_b}, the proposed model \dpautogm significantly thwarts the attacker's attempt to recover anything from the private data. However, similar with the results from model inversion attack, \dpvae is not robust enough to defend against this attack. We will also discuss in detail in Section~\ref{discuss}.

\subsection{Discussion}
\label{discuss}
Although both \dpvae and \dpautogm are differentially private generative models, the results show that \dpautogm is robust against all the attacks while \dpvae can only defend against the membership inference attack. The main difference between these two models is that \dpautogm uses the output of the encoder (a part of the autoencoder) as the generated data while \dpvae uses the output of the VAE. As the encoder functions can reduce the dimensions of the input data, we can envision that this operation will incur a big norm distance between the input data and the generated data in \dpautogm. 
Considering the model inversion attack and GAN attack both target at recovering part of the training data of a model, the best result on \dpautogm will be successfully recovering those encoded data while for \dpvae, the result will be recovering all. 
Therefore, it seems that the key to defend against these two attacks is not only differential privacy, but also the appearance of the generated data. This is also mentioned by Hitaj et al.~\cite{hitaj2017deep}, as they asserted that differential privacy is not effective in mitigating the developed GAN attack because differential privacy is not designed to solve such a problem. Differential privacy in deep learning targets at protecting the specific elements of training data, while the goal of these two attacks is to construct a data point which is similar to the training data. Even if the attacks are successful, differential privacy is not violated since the specific data points are not recovered.

\begin{table}[htbp]
\caption{Robustness of privacy preserving learning methods against different privacy attacks}
\label{tab:nearby}
\resizebox{\linewidth}{!}{
\centering
\begin{tabular}{c|c|c|c}
  \toprule
{Models} & {Model Inversion} & { Membership } &{GAN-based Attack against} \\
{} & {Attack}&{Inference Attack} &{Collaborative Deep Learning}\\
 \midrule
{\bf \dpautogm} & $\CIRCLE$ & $\CIRCLE$ & $\CIRCLE$ \\
{\bf \dpvae} & $\Circle$ & $\CIRCLE$ & $\Circle$\\
{\bf \dpdl} & $\Circle$ & $\CIRCLE$ & $\Circle$\\
{\bf \spate} & $\Circle$  & $\CIRCLE$ & $\Circle$\\
   \bottomrule
\end{tabular}}
  \centering {\small $\CIRCLE$: Robust~~$\Circle$: Not Robust}
\end{table}

\noindent \textbf{Remark.} Extensive experiments have shown that \dpautogm can mitigate all the three attacks. \dpvae is only robust against the membership inference attack (see Table~\ref{tab:nearby}).

\begin{figure*}[htb]
\centering
  \mbox{
    \subfloat[]{\includegraphics[height=.07\textwidth]{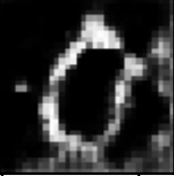}} \quad
    \subfloat[]{\includegraphics[height=.07\textwidth]{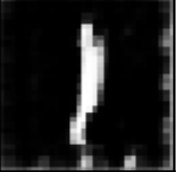}}  \quad
    \subfloat[]{\includegraphics[height=.07\textwidth]{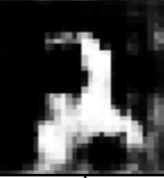}} \quad
    \subfloat[]{\includegraphics[height=.07\textwidth]{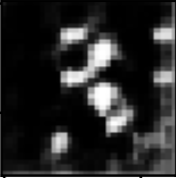}}   \quad
    \subfloat[]{\includegraphics[height=.07\textwidth]{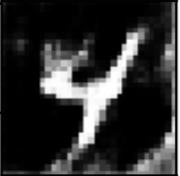}} \quad
     \subfloat[]{\includegraphics[height=.07\textwidth]{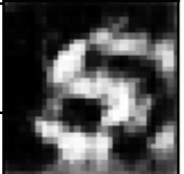}}  \quad
     \subfloat[]{\includegraphics[height=.07\textwidth]{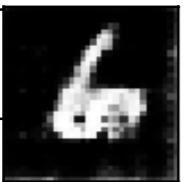}}   \quad
     \subfloat[]{\includegraphics[height=.07\textwidth]{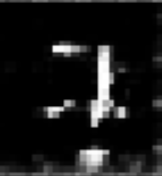}}   \quad
     \subfloat[]{\includegraphics[height=.07\textwidth]{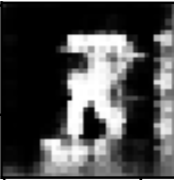}}   \quad
     \subfloat[]{\includegraphics[height=.07\textwidth]{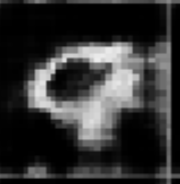}}   
     }
       \caption{With the original attack (images generated by the GAN-based attack against collaborative deep learning on the MNIST dataset)}
 \label{federated_picture_a}
   \vspace{-0.2cm}
\end{figure*}

\begin{figure*}[htb]
\centering
   \vspace{-0.5cm}
 \mbox{
      \subfloat[]{\includegraphics[height=.07\textwidth]{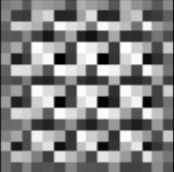}} \quad
     \subfloat[]{\includegraphics[height=.07\textwidth]{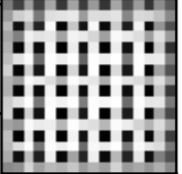}} \quad 
     \subfloat[]{\includegraphics[height=.07\textwidth]{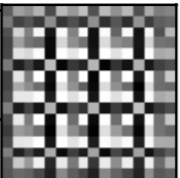}} \quad
     \subfloat[]{\includegraphics[height=.07\textwidth]{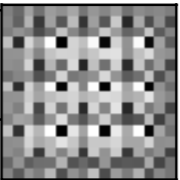}} \quad
     \subfloat[]{\includegraphics[height=.07\textwidth]{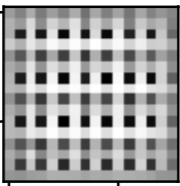}} \quad
     \subfloat[]{\includegraphics[height=.07\textwidth]{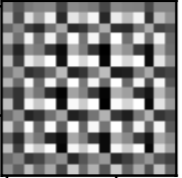}} \quad
    \subfloat[]{\includegraphics[height=.07\textwidth]{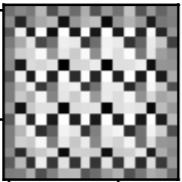}} \quad
     \subfloat[]{\includegraphics[height=.07\textwidth]{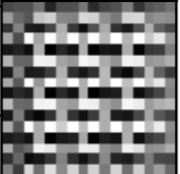}} \quad
    \subfloat[]{\includegraphics[height=.07\textwidth]{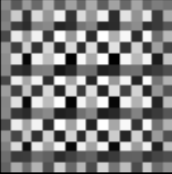}} \quad
     \subfloat[]{\includegraphics[height=.07\textwidth]{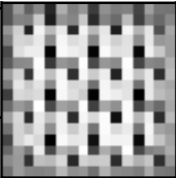}} 
     }
  \caption{With \dpautogm (images generated by the GAN-based attack against collaborative deep learning on the MNIST dataset)}
  \label{federated_picture_b}
    \vspace{-0.2cm}
\end{figure*}

\section{Deploying Data Generative Models on Real-World Applications}

To demonstrate the applicability of \dpautogm and \dpvae, we will show how they can be easily integrated with Machine Learning as a Service (MLaaS) commonly supported by major Internet companies and federated learning supported by Google. We integrate \dpautogm with both MLaaS and federated learning over all the datasets. We mainly focus on the utility performance of \dpautogm when integrated with federated learning, since federated learning is threatened by the GAN-based attack but can be effectively defended against by \dpautogm. We integrate \dpvae with MLaaS alone and evaluate it on the image dataset \mnist, as currently VAEs are usually used for generating images.

\subsection{Machine Learning as a Service}
MLaaS platforms are cloud-based systems that provide simple APIs as a web service for users who are interested in training and querying machine learning models. For a given task, a user first submits the private data through a web page interface or an mobile application created by developers, and selects the features for the task. Next, the user chooses a machine learning model from the platform, tunes the parameters of the model, and finally obtains the trained model. All these processes can be completed inside the mobile application. However, the private data submitted by innocent users can be maliciously exploited if the platform is compromised, which raises serious privacy concerns. In this paper, our \dpautogm and \dpvae can serve as a data privacy protection module to protect privacy of the private data. To this end, users can first build \dpautogm or \dpvae locally, train the generative models with the private data, and then upload the generated data for later training. As we will show in the experiment, this will incur negligible utility loss for training, while significantly protecting data privacy. With \dpautogm and \dpvae, even if these platforms are compromised, the privacy of sensitive data can still be preserved. In addition, we will show that training \dpautogm and \dpvae locally requires only a few computational resources.

When applying the proposed \dpautogm and \dpvae to MLaaS, we choose to examine three mainstream MLaaS platforms, which are Google Prediction API~\cite{google}, Amazon Machine Learning~\cite{amazon}, and Microsoft Azure Machine Learning~\cite{microsoft}. 
We then set the differential privacy budget $\epsilon$ and $\delta$ to be $1$ and $10^{-5}$, respectively, for \dpvae and \dpautogm. Similar with the evaluation section, we regard all the training data as private data and for \dpautogm, we split the test data the same way as we do in Section~\ref{autogm}. As we can see from Figure~\ref{accuracy_application}, using the generated data by \dpautogm for training, we can achieve comparatively high accuracy (accuracy deteriorating within 8\%) on all three platforms for all datasets. Strikingly, we find that the model trained with generated data sometimes even outperforms the one trained with original data (see trained models on Amazon over \mnist). For \dpvae, the result is shown in Figure~\ref{service_mnist}. We can see that \dpvae can achieve comparable utility (accuracy deteriorating within 3\%) on all the three platforms on \mnist. This clearly shows that \dpvae and \dpautogm have the potential to be well integrated into MLaaS platforms and provide privacy guarantees for users' private data and retain high data utility at the same time.

Furthermore, we show the time cost of training \dpautogm (58.2s) and \dpvae (27.9s) under 10 epochs on \mnist dataset. 
The evaluation is done with Intel Xeon CPU with 2.6GHZ, GPU of GeForce GTX 680, Ubuntu 14.04 and Tensorflow. 
Most recently, Tensorflow Mobile~\cite{mobile} has been proposed to deploy machine learning algorithms on mobile devices. We, therefore, believe it will cost much less to train such generative models locally on mobile devices.

\begin{figure*}[t]
  \centering
  \mbox{
  \subfloat[MNIST (MLaaS)\label{service_mnist}]{\includegraphics[width=0.18\textwidth]{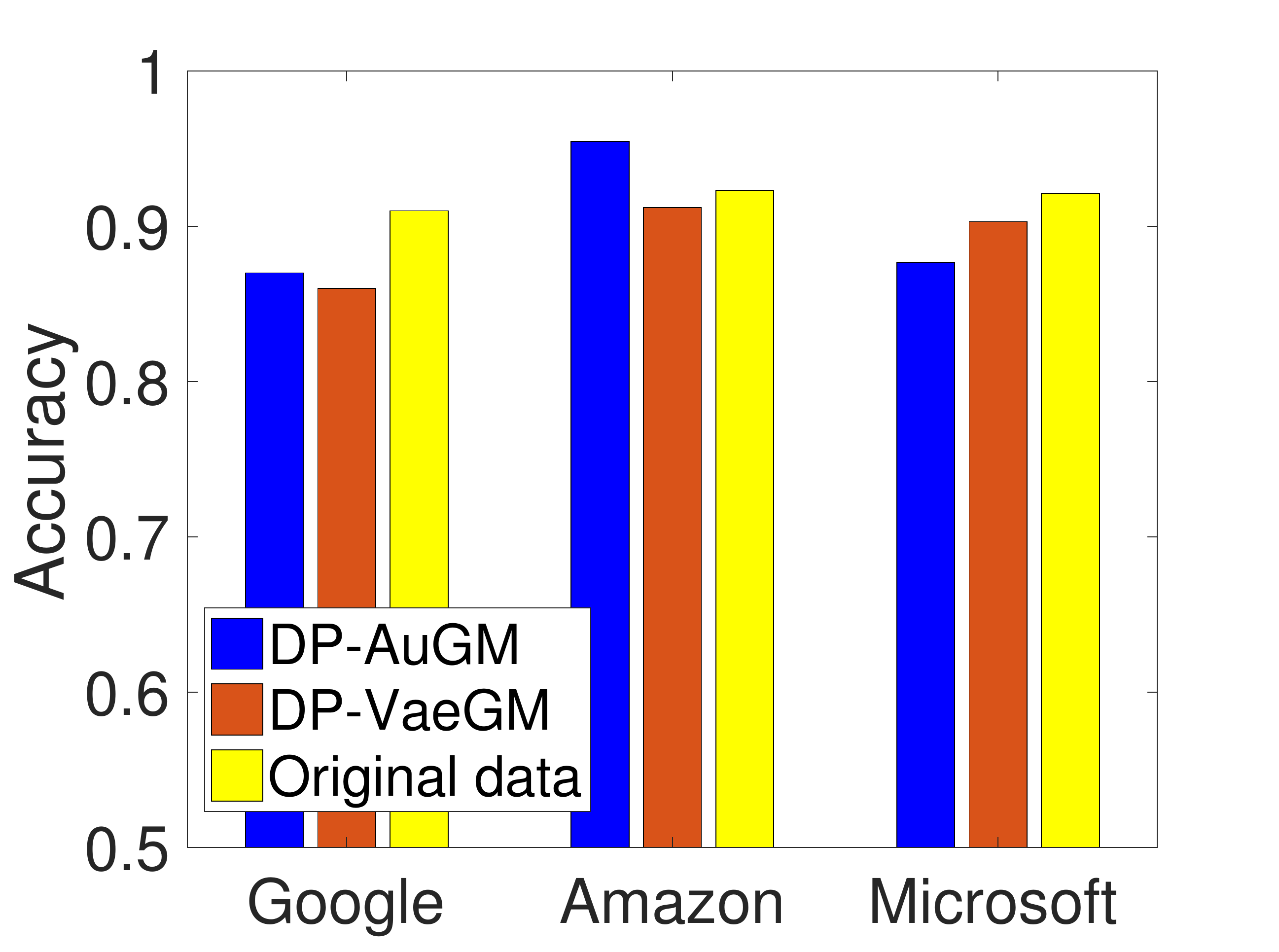}}
\quad
   \subfloat[\adult (MLaaS)\label{service1}]{\includegraphics[width=0.18\textwidth]{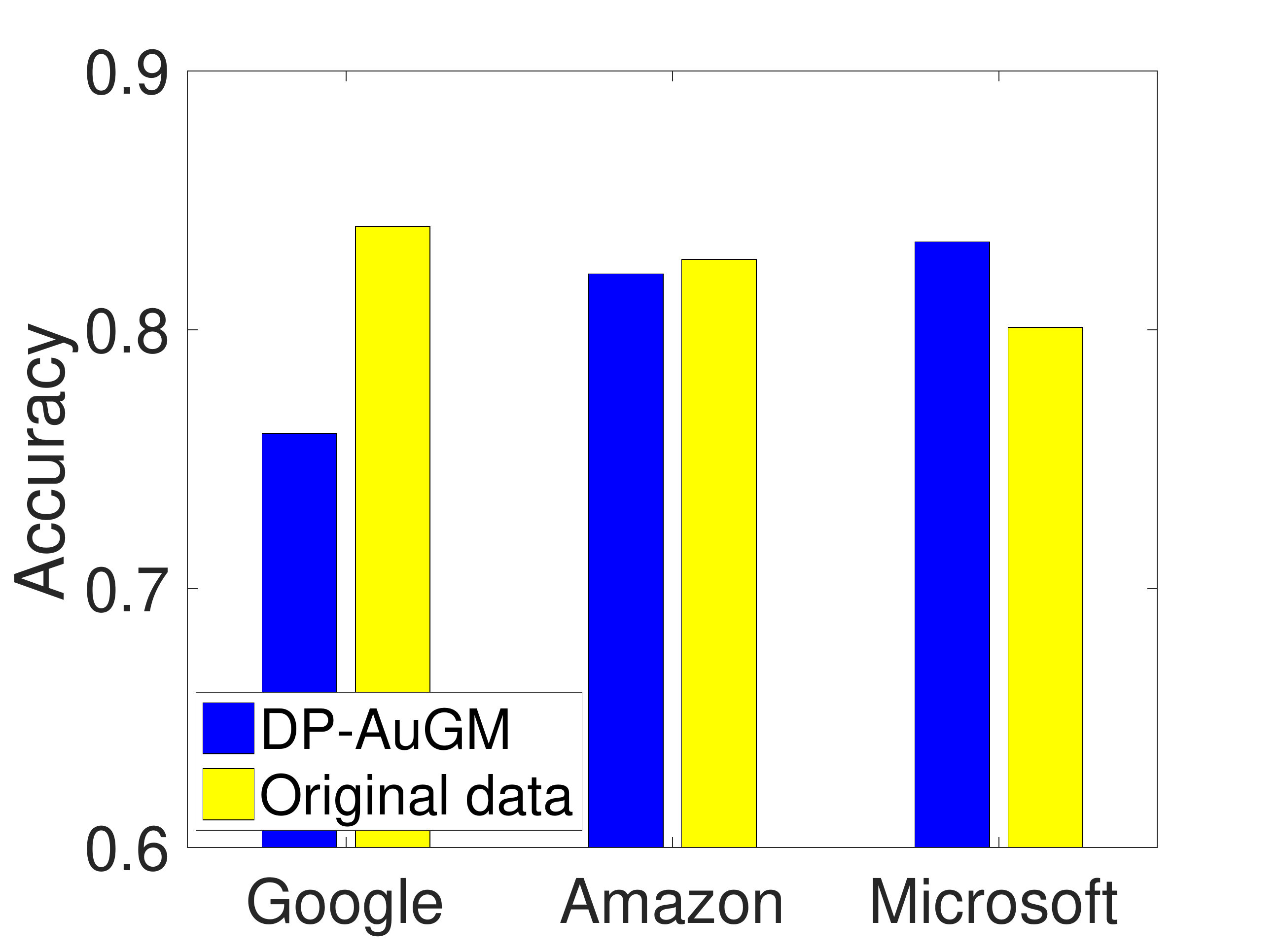}}
\quad
     \subfloat[\hospital (MLaaS)\label{service2}]{\includegraphics[width=0.18\textwidth]{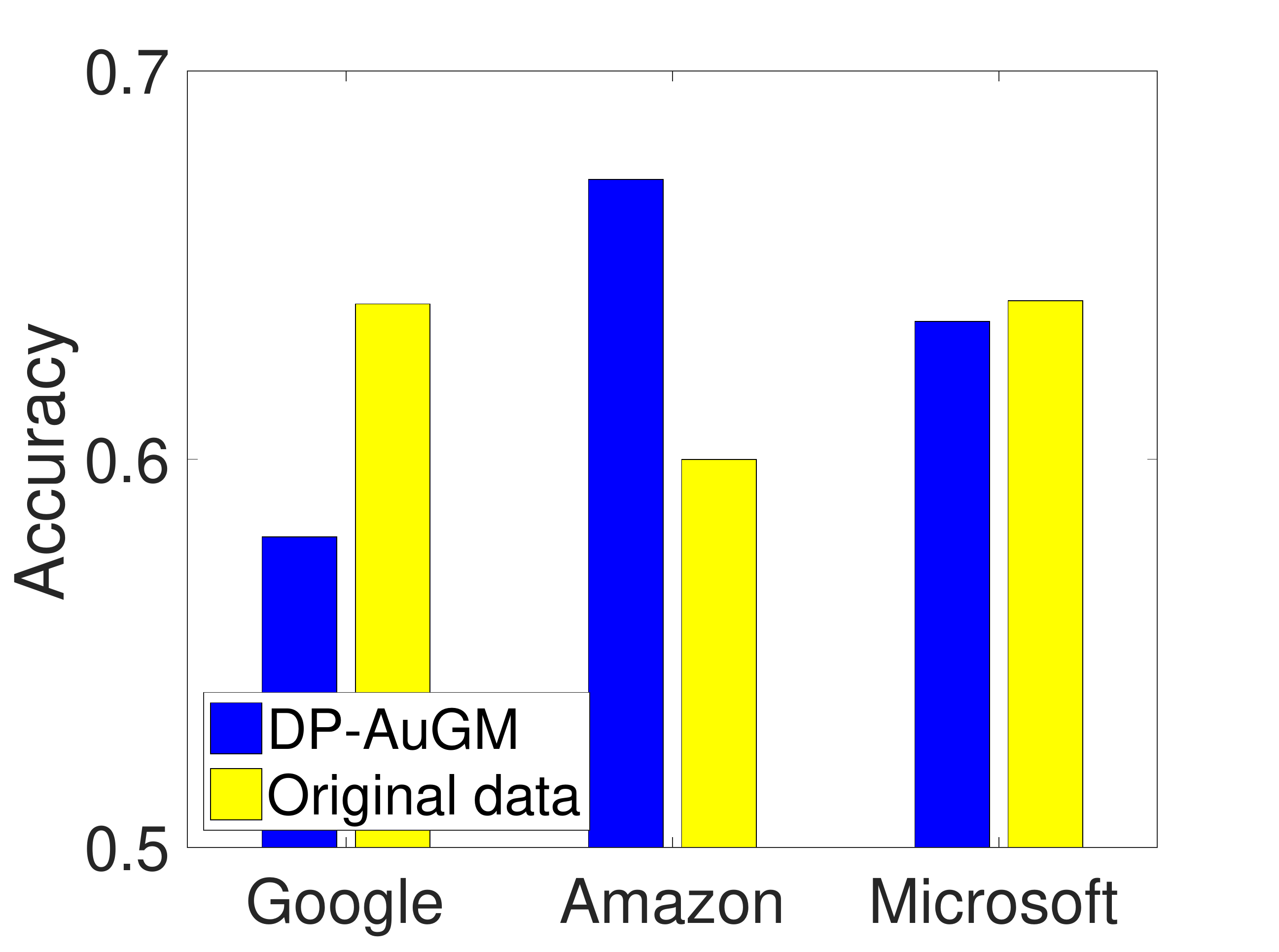}}
\quad
     \subfloat[\malware (MLaaS)\label{service3}]{\includegraphics[width=0.18\textwidth]{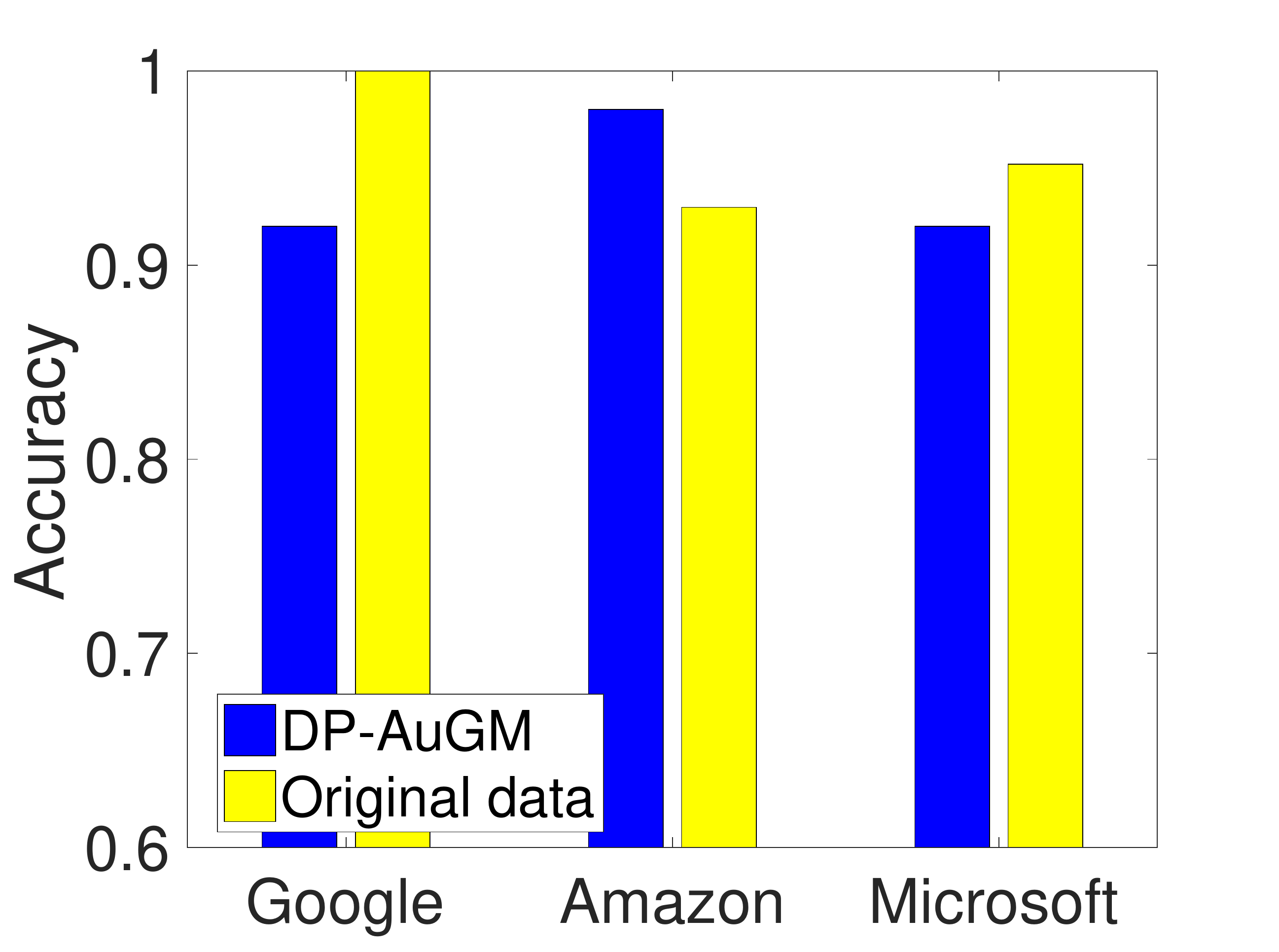}}
  \quad
            \subfloat[Accuracy on four datasets (federated learning) \label{service_fe}]{\includegraphics[width=0.18\textwidth]{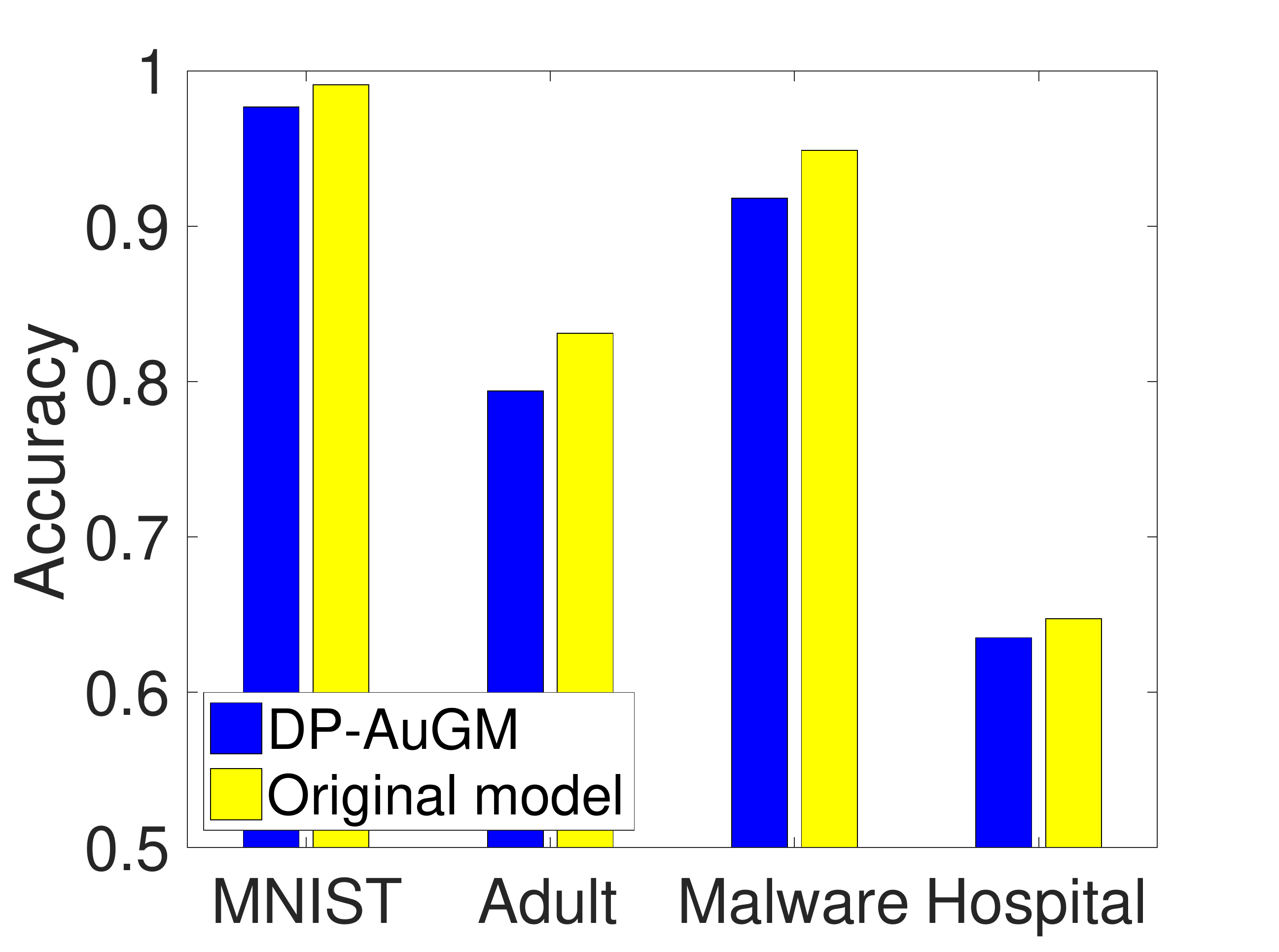}}
            }
  \caption{Accuracy of trained models when integrating proposed generative models with MLaaS and federated learning}
  \label{accuracy_application}
  \vspace{-0.2cm}
\end{figure*}

\subsection{Federated Learning}
Federated learning~\cite{mcmahan2017federated}, which is proposed by Google, enables mobile users to collaboratively train a shared prediction model and keep all their distributed training data local. Users typically train the model locally on their own device, upload the summarized parameters as a small focused update, and download the parameters averaged with other users' updates collaboratively using secure multiparty computation (MPC), without needing to share their personal training data in the cloud.

Federated learning is demonstrated to be private since the individual users' data is stored locally and the updates are securely aggregated by leveraging MPC to compute model parameters. However, the recent paper~\cite{hitaj2017deep} declares that federated learning is secure only if we consider the attacker is the cloud provider who scrutinizes individual updates. If the attackers are the casual colluding participants, private data of one participant can still be recovered by other users who aim to attack. Hitaj et al.~\cite{hitaj2017deep} have shown that only applying differential privacy in federated learning is not sufficient to mitigate the GAN-based attack, and a malicious user is able to successfully recover private data of others. 

In Section~\ref{sec:attacks}, we show that \dpautogm is robust enough to mitigate the GAN attack. Thus, in this section, we will mainly consider whether \dpautogm can be well integrated into the federated learning to protect privacy and retain high data utility.
We show the concrete steps toward integrating \dpautogm as below. Note that the first two steps are added to the original federated learning platform.
\begin{enumerate}[wide=1pt,leftmargin=10pt]
\item \textit{Users first train \dpautogm locally with the private data.}
\item \textit{After training \dpautogm, users use \dpautogm and public data to generate new training data.}
\item Users train the local model with \textit{generated data} locally and upload the summarized parameters to the server.
\end{enumerate}
Next we will empirically show that \dpautogm can be well integrated into federated learning over four datasets. 

\noindent \textbf{Settings.}
The structure of an autoencoder and differential privacy parameters can be specified by a central server such as Google, and will be publicly available to any user. As a proof of concept, we hereby set the differential privacy parameters~$\epsilon$ and $\delta$ to be $1$ and $10^{-5}$, respectively. For each user in the federated learning, we evenly split the private data and public data for usage.

\noindent \textbf{Hyper-parameters.}
We set the default learning rate to be 0.001, the batch size to be 100, the number of users to be 10, and the uploading fraction to be 0.1. We will also test how \dpautogm performs across different parameters later.

\noindent \textbf{In Comparison with the Original Federated Learning.}
We apply \dpautogm to federated learning and compare it with the original setting without \dpautogm. As we can see from Figure~\ref{service_fe}, after we add \dpautogm model to the pipeline, the accuracy drops only within 5\% for all datasets. Hence, it shows the proposed \dpautogm can be well integrated into federated learning without affecting its utility too much. In the following part, we study in detail about the model sensitivity on the \mnist dataset.

\noindent \textbf{Effect of Other Parameters.}
We further examine the effect of the number of users and the upload fraction over the differentially private federated learning model.

\noindent \textit{\textbf{(a) Number of Users.}} We choose the number of users to be 10, 20, and 40. From Figure~\ref{fede_number}, we can see the difference in number of users will only affect the speed of convergence a bit without affecting the final data utility. We find that although more users will take slightly more time for the model to converge, the accuracy of the differentially private model actually converges to the same result within 50 epochs.

\noindent \textit{\textbf{(b) Upload Fraction.}} We choose the upload fraction as 0.001, 0.01, and 0.1 to analyze the proposed method. As we can see from Figure~\ref{fede_upload}, different learning rates only have negligible impacts on the trained model.

\begin{figure*}[t]
  \centering
  \mbox{
  \subfloat[Accuracy on \mnist with different number of users \label{fede_number}]{\includegraphics[width=0.4\textwidth]{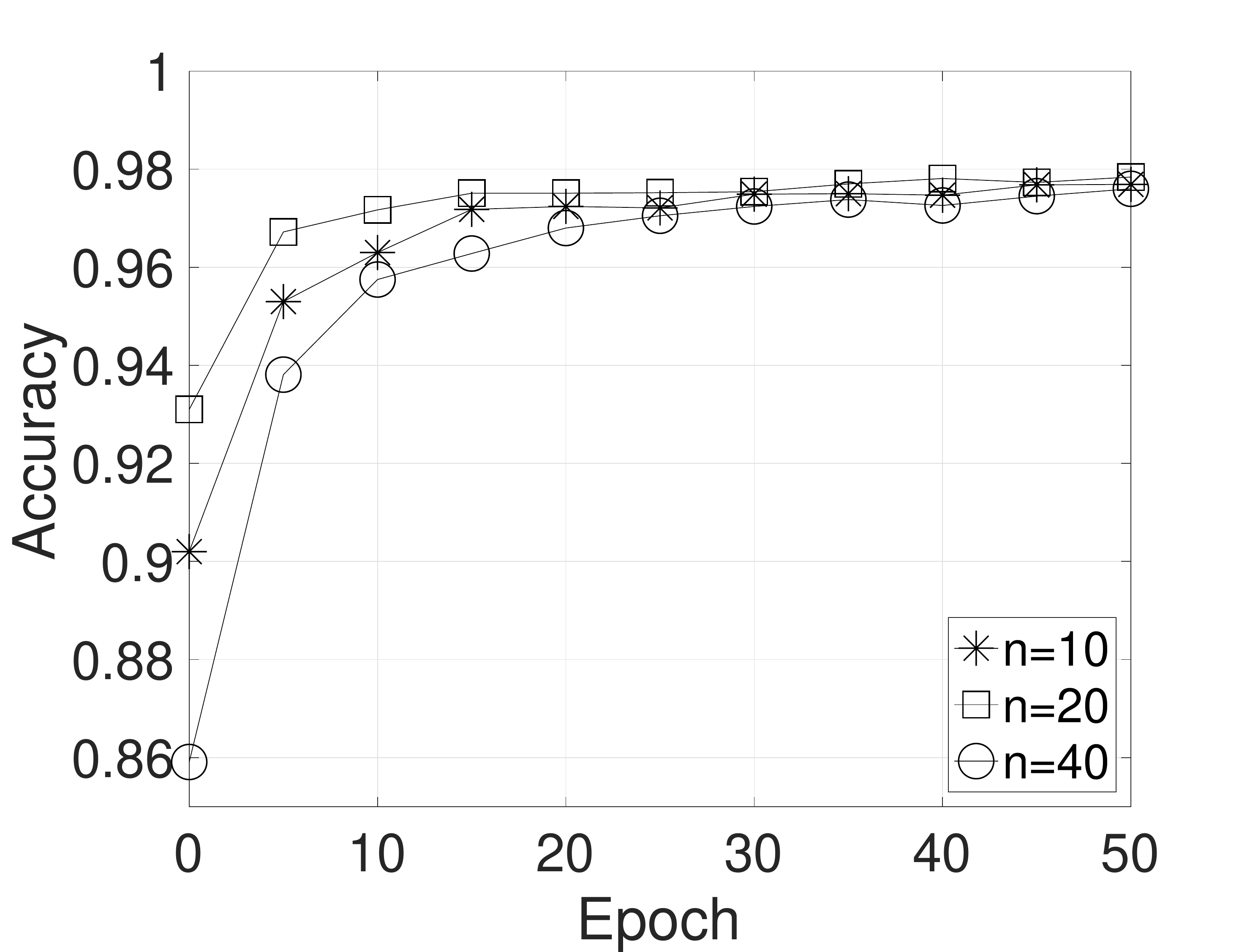}}
    \quad
  \subfloat[Accuracy on \mnist across different upload fraction \label{fede_upload}]{\includegraphics[width=0.4\textwidth]{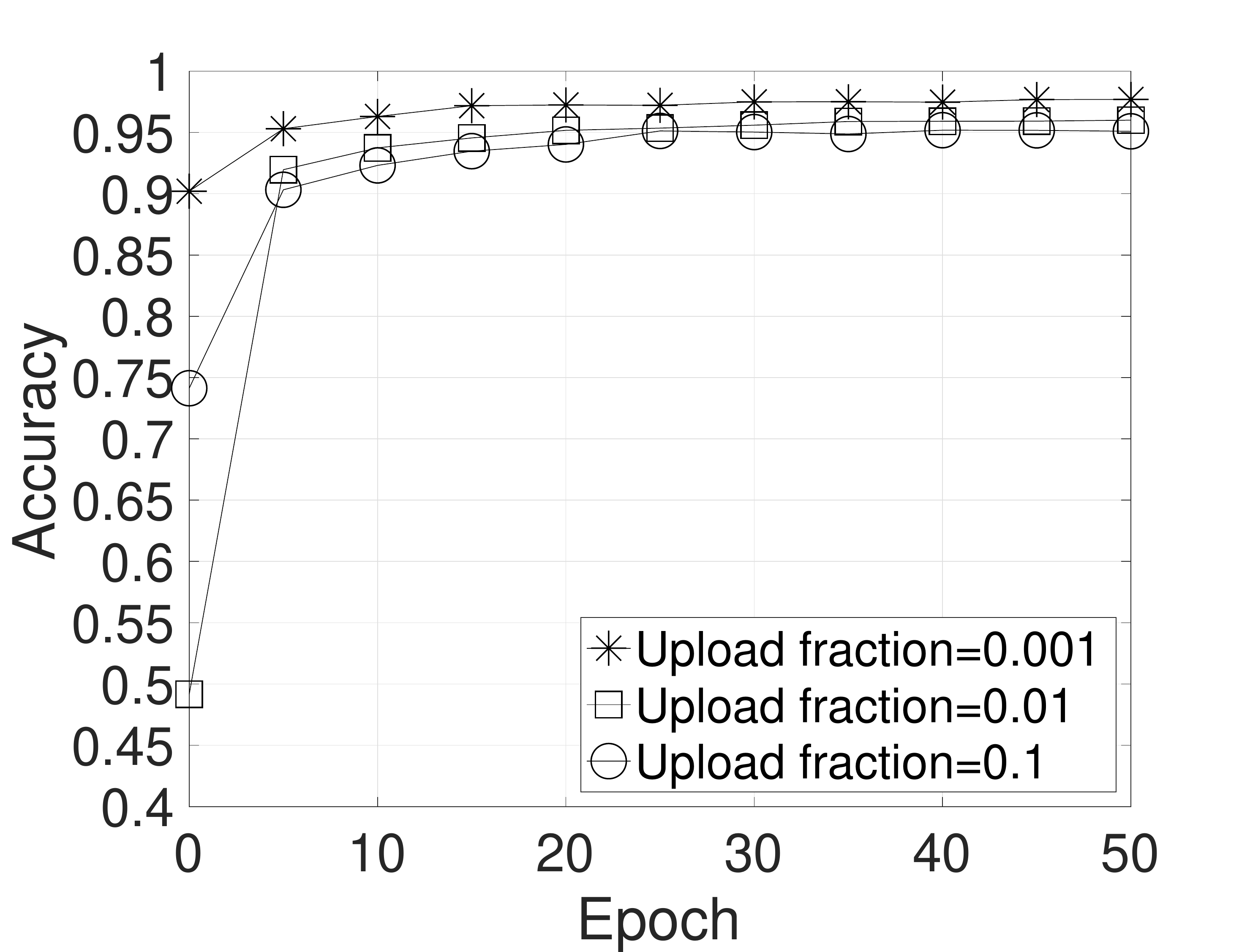}}
    }
  \caption{The performance of federated learning integrated \dpautogm under different hyper-parameters}
  \label{hyper}
  \vspace{-0.2cm}
\end{figure*}

\noindent \textbf{Remark.} We have shown that \dpautogm can be well integrated with MLaaS and federated learning, and \dpvae can be well integrated with MLaaS. The integrated models can protect privacy and preserve high data utility at the same time.

\section{Related Work}

\subsection{Privacy Attacks on Machine Learning Models}

Specifically, Homer et al.~\cite{homer2008resolving} show that it is possible to learn whether a target individual was related to certain disease by comparing the target's profile against the aggregated information obtained from public sources. This attack was then extended by Wang et al.~\cite{wang2009learning} by performing correlation attacks, without prior knowledge about the target. Backes et al.~\cite{backes2016membership} propose to conduct the membership inference attack against individuals contributing their microRNA expressions to scientific studies. If an attacker can learn information about individual's genome expression, he can potentially infer/profile the victim's future/historical health records, which can lead to severe consequences. Shokri and Shmatikov~\cite{shokri2016membership} later show that machine learning models can leak information about medical data records by performing membership attack against well trained models. Recently, Hitaj et al.~\cite{hitaj2017deep} show that a GAN-based attack can compromise user privacy in the collaborative learning setting~\cite{shokri2015privacy}, where each participant collaboratively trains his or her own model with private data locally. Hitaj et al.~\cite{hitaj2017deep} also warn that simply adding differentially private noise is not robust enough to mitigate the attack. In addition, Hayes et al.~\cite{hayes2018logan} investigate the membership inference attack for generative models by using GANs~\cite{goodfellow2014generative} to detect overfitting and recognize training inputs. More recently, Liu et al.~\cite{liu2018generative} define the co-membership inference attack against generative models. 

Given these existing privacy attacks, learning with generated data from DP generative models can potentially defend against them, such as the representative model inversion attack, membership inference attack, and GAN-based attack against collaborative deep learning. To the best of our knowledge, the learning method that can defend against all these attacks has not been proposed or systematically examined before.

\subsection{Differentially Private Learning Methods}

The goal of differentially private learning models is to protect sensitive information of individuals within the training set. Differential privacy is a strong and common notion to protect the data privacy~\cite{dwork2014algorithmic}. 
Differential privacy can also be used to mitigate membership inference attacks, as its indistinguishability-based definition formally proves that the presence or absence of an instance does not affect the output of the learned model significantly~\cite{shokri2016membership}. 
A common approach to achieving differential privacy is to add noise from Laplacian~\cite{dwork2008differential} or Gaussian distribution~\cite{dwork2006our} whose variance is determined by the privacy budget. In practice, differentially private schemes are often tailored to the spatio-temporal location privacy analysis~\cite{machanavajjhala2008privacy,  rastogi2010differentially, shokri2012protecting, acs2014case, to2016differentially}.

To protect the privacy of machine learning models, random noise can be injected to input, output, and objectives of the models. 
Erlingsson et al.~\cite{erlingsson2014rappor} propose to randomize the input and show that the randomized input still allows data collectors to gather meaningful statistics for training. Chaudhuri et al.~\cite{chaudhuri2011differentially} show that by adding noise to the cost function minimized during learning, $\epsilon$-differential privacy can be achieved. 
In terms of perturbing objectives, Shokri et al.~\cite{shokri2015privacy} show that deep neural networks can be trained with multi-party computations from perturbed model parameters to achieve differential privacy guarantees.
Deep learning with differential privacy is proposed~\cite{abadi2016deep} by adding noise to the gradient during each iteration. 
They further use moment accountant to keep track of the spent privacy budget during the training phase. 
However, the prediction accuracy of the deep learning system will degrade more than 13\% over the CIFAR-10 dataset when large differential privacy noise is added~\cite{abadi2016deep}, which is unacceptable in many real-world applications where high prediction accuracy is pursued, such as autonomous driving~\cite{geiger2012we} and face recognition~\cite{graham1998characterising}. This is also aligned with the warning proposed by Hitaj et al.~\cite{hitaj2017deep} that using differential privacy to provide strong privacy guarantees cannot be applied to all scenarios, especially where the GAN-based attack can be applied. 
Later, private aggregation of teacher
ensembles (PATE) has been proposed, which first learns an ensemble of teacher models on a disjoint subset of training data, and aggregates the output of these teacher models to train a differentially private student model for prediction~\cite{papernot2016semi}. The queries performed on the teacher models are designed to minimize the privacy cost of these queries. Once the student models are trained, the teacher models can be discarded. PATE is within the scope of knowledge aggregation and transfer for privacy~\cite{pathak2010multiparty, hamm2016learning}.
An improved version of PATE, scalable PATE, is proposed by introducing new aggregation algorithm to achieve better data utility~\cite{papernot2018scalable}. 

At inference, random noise can also be introduced to the output to protect privacy. However, this severely decays the test accuracy, because the amount of noise introduced increases with the number of inference queries answered by the machine learning model. Note that homomorphic encryption~\cite{gilad2016cryptonets} can also be applied to protect the confidentiality of each individual input. The main limitations are the performance overhead and the restricted set of arithmetic operations supported by homomorphic encryption.

Various approaches have been proposed for the automatic discovery of sensitive entities, such as identifiers, and redact them to protect privacy. The simplest of these rely on a large collection of rules, dictionaries, and regular expressions (e.g., \cite{Beckwith:06,sweeney1996replacing}). Chakaravarthy et al.~\cite{chakaravarthy2008efficient} proposed an automated data sanitization algorithm aimed at removing sensitive identifiers while inducing the least distortion to the contents of documents. However, this algorithm assumes that sensitive entities, as well as any possible related entities, have already been labeled. Similarly, Jiang et al.~\cite{jiang2009t}~have developed the $t$-plausibility algorithm to replace the known (labeled) sensitive identifiers within the documents and guarantee that the sanitized document is associated with at least $t$ documents. Li et al.~\cite{li2017scalable} have proposed a game theoretic framework for automatic redacting sensitive information. In general, finding and redacting sensitive information with high accuracy is still challenging.

In addition, there has been recent work on making generative neural networks differentially private~\cite{acs2018differentially}. It achieved their differentially private generative models on VAEs by using differentially private kernel $k$-means and differentially private gradient descent. Different from their work, we realize differentially private generative models on both autoencoders and VAEs. We further show that our proposed methods can mitigate realistic privacy attacks and can  seamlessly be applied to real-world applications.

In general, unlike previously proposed techniques, our proposed differentially private generative models can guarantee differential privacy while maintaining data utility for learning tasks. Our proposed models achieve all three goals: protect privacy of training data; enable users to locally customize the privacy preference by configuring the generative models; retain high data utility for generated data. The proposed models achieve these goals at a much lower computation cost than aforementioned differentially private mechanisms and cryptographic techniques, such as secure multi-party computation or homomorphic encryption~\cite{gilad2016cryptonets}. Our generative models can also seamlessly be integrated with MLaaS and federated learning in practice.

\section{Conclusion}
We have designed, implemented, and evaluated two differentially private data generative models---a differentially private autoencoder-based generative model (\dpautogm) and a differentially private variational autoencoder-based generative model (\dpvae). We show that both models can provide strong privacy guarantees and retain high data utility for machine learning tasks. We empirically demonstrate that \dpautogm is robust against the model inversion attack, membership inference attack, and GAN-based attack against collaborative deep learning, and \dpvae is robust against the membership inference attack. We conjecture that  the  key  to  defend against model inversion and GAN-based attacks is to distort the training data while differential privacy is targeted to protect membership privacy.
Furthermore, we show that the proposed generative models can be easily integrated with two real-world applications---machine learning as a service and federated learning, which are previously threatened by the membership inference attack and GAN-based attack, respectively.  We demonstrate that the integrated system can both protect privacy of users' data and retain high data utility.

Through the study of privacy attacks and corresponding defensive methods, we here emphasize that it is important to generate differentially private synthetic data for various machine learning systems to secure current learning tasks. As we are the first to propose differentially private data generative models that can defend against the contemporary privacy violation attacks, we hope that our work will help pave the way toward designing more effective differentially private learning methods.

\bibliographystyle{abbrv}
\bibliography{PETS_19}

\clearpage
\section*{Appendix}
\appendix

\section{Model Architectures}\label{sec: model_arc}

\begin{table}[tbhp]
\caption{Model structures of \dpautogm over different datasets}
\resizebox{\linewidth}{!}{
\centering
\begin{tabular}{c|c|c|c}
  \toprule
  {\bf MNIST} & {\bf Adult Census Data} & {\bf Texas Hospital Stays Data} & {\bf Malware Data}\\
  \midrule
{FC(400)+Sigmoid} & {FC(6)+Sigmoid} & {FC(400)+Sigmoid} & {FC(50)+Sigmoid} \\
{FC(256)+Sigmoid} & {FC(100)+Sigmoid} & {FC(776)+Sigmoid} & {FC(142)+Sigmoid} \\
{FC(400)+Sigmoid} & {} & {} & {} \\
{FC(784)+Sigmoid} & {} & {} & {} \\
\bottomrule
\end{tabular}}
\end{table}

\begin{table}[htbp]
\caption{Model structures of \dpvae over MNIST}
\centering
\begin{tabular}{c}
  \toprule
  {\bf MNIST}\\
  \midrule
{FC(500)+Sigmoid}\\
{FC(500)+Sigmoid}\\
{FC(20)+Sigmoid ; FC(20)+Sigmoid}\\
{Sampling Vector(20)}\\
{FC(500)+Sigmoid}\\
{FC(500)+Sigmoid}\\
{FC(784)+Sigmoid}\\
\bottomrule
\end{tabular}
\end{table}

\begin{table}[htbp]
\caption{Structures of machine learning models over different datasets with \dpautogm}
\resizebox{\linewidth}{!}{
\centering
\begin{tabular}{c|c|c|c}
  \toprule
  {\bf MNIST} & {\bf Adult Census Data} & {\bf Texas Hospital Stays Data} & {\bf Malware Data}\\
  \midrule
{Conv(5x5,1,32)+Relu} & {FC(16)+Relu} & {FC(200)+Relu} & {FC(4)+Relu} \\
{MaxPooling(2x2,2,2)} & {FC(16)+Relu} & {FC(100)+Relu} & {FC(3)+Relu} \\
{Conv(5x5,32,64)+Relu} & {FC(2)} & {FC(10)} & {FC(2)} \\
{MaxPooling(2x2,2,2)} & {} & {} & {} \\
{Reshape(4x4x64)} & {} & {} & {} \\
{FC(10)} & {} & {} & {} \\
\bottomrule
\end{tabular}}
\end{table}

\begin{table}[h]
\caption{Structures of machine learning models over different datasets with \dpvae}
\centering
\begin{tabular}{c}
  \toprule
  {\bf MNIST}\\
  \midrule
{Conv(5x5,1,32)+Relu}  \\
{MaxPooling(2x2,2,2)}  \\
{Conv(5x5,32,64)+Relu} \\
{MaxPooling(2x2,2,2)}\\
{Reshape(7x7x64)}\\
{FC(1024)}\\
{FC(10)} \\
\bottomrule
\end{tabular}
\end{table}
\end{document}